%% file: main.tex
\documentclass[12pt]{article}

\input{_preamble}

\newtheorem{theorem}{Theorem}
\newtheorem{lemma}{Lemma}
\title{Asymptotically Efficient Data-adaptive Penalized Shrinkage Estimation with Application to Causal Inference}

\newcommand{\arxiv}{1}

\ifx\arxiv\undefined 
    \date{}
\fi
\ifx\anonymized\undefined
    \author[1,*]{Herbert P. Susmann}
    \author[2]{Yiting Li}
\author[2]{Mara A. McAdams-DeMarco}
    \author[1]{Wenbo Wu}
    \author[1]{Iv\'an D\'iaz}

    \date{}
    
    \affil[1]{\small Division of Biostatistics, Department of Population
      Health, New York University Grossman School of Medicine, New York, NY, USA}
      \affil[2]{\small Department of Surgery, NYU Grossman School of Medicine, USA}

    \affil[*]{Corresponding author: susmah01@nyu.edu}
\else
    \author{}
\fi

\begin{document}
\maketitle

\begin{abstract}
    A rich literature exists on constructing non-parametric estimators with optimal asymptotic properties. In addition to asymptotic guarantees, it is often of interest to design estimators with desirable finite-sample properties; such as reduced mean-squared error of a large set of parameters. We provide examples drawn from causal inference where this may be the case, such as estimating a large number of group-specific treatment effects. We show how finite-sample properties of non-parametric estimators, particularly their variance, can be improved by careful application of \textit{penalization}. Given a target parameter of interest we derive a novel penalized parameter defined as the solution to an optimization problem that balances fidelity to the original parameter against a penalty term. By deriving the non-parametric efficiency bound for the penalized parameter, we are able to propose simple data-adaptive choices for the $L_1$ and $L_2$ tuning parameters designed to minimize finite-sample mean-squared error while preserving optimal asymptotic properties. The $L_1$ and $L_2$ penalization amounts to an adjustment that can be performed as a post-processing step applied to any asymptotically normal and efficient estimator. We show in extensive simulations that this adjustment yields estimators with lower MSE than the unpenalized estimators. Finally, we apply our approach to estimate provider quality measures  of kidney dialysis providers within a causal inference framework.
\end{abstract}
\keywords{causal inference; doubly robust estimation; penalization; shrinkage estimator}

\section{Introduction}
In many settings it is of interest to define and estimate a large set of related statistical parameters. This is often the case in causal inference, where one may wish to estimate a large set of related treatment effects. For example, in studies of an intervention applied in multiple sites, one may wish to estimate both the average effect of the intervention marginally across all sites as well as the average effect within each site; here, there are as many statistical parameters as there are sites. When there are many sites, estimating the site-specific effects may be challenging; this is especially true when there are sites with few data. Another salient example arises in healthcare provider profiling applications, in which many healthcare providers are evaluated based on their patient outcomes. A more general example is determining the importance of a large number of variables in a prediction model, which may involve estimating a large number of variable importance measures \citep{Williamson2021variableimportance}. 

When estimating a set of statistical parameters in real-world scenarios there is not typically sufficient mechanistic knowledge to justify the use of parametric models. Non-parametric, data-adaptive approaches are instead warranted. For example, the relationship between patient health outcomes, patient characteristics, and healthcare provider characteristics is highly complex, and cannot be accurately described by a simple (e.g. linear) relationship between variables. In order to avoid such strong assumptions, we prefer to work within a non-parametric framework in which we seek to estimate low-dimensional statistical summaries, such as a set of treatment effects, of an infinite-dimensional nuisance parameter, such as the set of all probability laws defined on the support of the data.

We guide the development of our estimators using semi-parametric efficiency theory, which characterizes lower bounds on the asymptotic performance of non-parametric estimators. Based on foundational work by H\'ajek and Le Cam \citep{hajek1970,hajek1972,lecam1972} and further developed by \cite{pfanzagl1982contributions,vdv1992,Bickel97}, among others (see \citealt[Chapter 25]{vanderVaart98} for an overview), this theory extends classical efficiency results for finite-dimensional parameters of smooth parametric models to the functionals of non-parametric, infinite-dimensional nuisance parameters. A key result is the convolution theorem, which establishes that the optimal limiting distribution for regular non-parametric estimators is gaussian with covariance determined by the \textit{efficient influence function} (EIF) of the functional. The EIF plays a similar role as the Fisher information for parametric models, which characterizes the parametric efficiency bound through the Cramer-Rao theorem. Thus, characterizing the form of the EIF for a statistical functional is a key task, as it characterizes the efficiency bound for estimating the functional in a non-parametric model.

Remarkably, several non-parametric estimation strategies have been developed for constructing non-parametric estimators that achieve the semi-parametric efficiency bound; these include one-step estimation, targeted maximum likelihood estimation, and estimating equations, among others \citep{pfanzagl1982contributions,Bickel97, Tsiatis06,vanderLaanRubin2006tmle} (see \citealt{kennedy2024review} for an accessible review). These estimators are typically built using the form of the EIF for the target statistical functional. Thus, deriving the EIF is useful for another reason: it both characterizes the efficiency bound, and provides a path towards constructing estimators that achieve this bound.

Semi-parametric efficiency theory, including the convolution theorem, provides an asymptotic theory of optimality for non-parametric estimators. However, we may wish to design estimators with additional finite-sample properties. For example, it may be desirable to find an estimator for a set of parameters for which each individual estimator may be \textit{biased} in finite samples, yet the \textit{mean-squared error} defined jointly over the set of parameters is lower. A related goal may be to find an estimator that has lower joint finite-sample mean-squared error and simultaneously summarizes the parameters in a useful way, for example by introducing \textit{sparsity}. That is, it is often desirable to have estimates that are not ``meaningfully far from zero'' shrunk identically to zero (where what it means to be ``meaningfully far from zero'' requires careful formalization.) Ideally, an estimator would have these finite-sample properties while still achieving the asymptotic optimality given by the convolution theorem, in which the limiting distribution is gaussian with variance given by the variance of the EIF.

In this paper, we investigate how \textit{penalization} can be used to construct alternative estimators with useful finite-sample properties, such as improved finite-sample variance and sparsity, while nonetheless having optimal asymptotic properties. First, we propose a general theoretical framework for defining penalized parameters. Our framework defines a penalized parameter as the solution to an optimization problem that balances fidelity to the original parameter (as measured via an arbitrary loss function) and an arbitrary penalization term. Our framework therefore encompasses penalized parameters defined using squared-error loss functions and $L_2$ and $L_1$ penalties, aping Ridge and Lasso regression, respectively. In practice, we allow the degree of penalization to depend on the sample size, with the goal that as sample size goes to infinity the penalized parameter converges to the original parameter. The penalized estimator therefore inherits the favorable asymptotic properties of the original estimator. We provide three examples to illustrate our proposals. First, we examine a non-parametric linear association parameter with which we directly compare our approach to traditional penalized regression methods. Second, we use as further examples two causal inference parameters: group-specific average treatment effects and indirectly standardized outcomes.  

Next, we apply tools from semi-parametric efficiency theory to derive a general form for the efficient influence function (EIF) of the penalized parameter. The EIF characterizes the efficiency bound of semi-parametric estimators of the penalized parameters. Knowledge of the efficiency bound allows us to derive data-adaptive choices of the penalization tuning parameters in the $L_2$ and $L_1$ cases. Under these data-adaptive choices, for which the degree of penalization depends on the sample size, we show that as sample size increases the EIF of the penalized parameter converges to the EIF of the original parameter. Thus, asymptotically our estimator recover the same limiting properties of non-penalized non-parametric estimator. Furthermore, the asymptotic results lead to construction of asymptotically valid statistical inference on the original target parameter of interest, including the construction of  confidence intervals. As such, our method amounts to a finite-sample correction of the point estimate designed to yield lower variance at the cost of introducing (finite-sample) bias. Practically speaking, we show that this penalization procedure can be applied as a post-processing step to the estimates yielded by any asymptotically normal and efficient estimator of the target parameter. This makes our methods easily applicable to the outputs provided by standard statistical software. 

Our approach is illustrated using simulated data in Figure~\ref{fig:group-effects-example}. Data are simulated for a trial of an intervention applied in multiple groups; for example, it could be a treatment intervention in multiple hospitals. The simulation is designed such that the true treatment effect for each group is uniformly distributed between -1 and 1. For three simulated datasets with increasing sample sizes, we first applied a doubly robust targeted causal inference estimator separately within each group to estimate the group-specific treatment effects (referred to as \textit{no penalty} in the figure). We then applied our proposed methods to estimate $L_1$ and $L_2$-penalized group-specific treatment effects. At the smallest sample sizes, the penalized estimates are shrunk towards zero, which improves the quality of the estimates. The $L_1$-penalized estimates are in some cases shrunk exactly to zero. Due to the data-adaptive choice of the penalization parameter, as the sample size increases the unpenalized and penalized point estimates (and confidence intervals) converge to each other; as such, the penalized estimates inherit the optimal asymptotic properties of the doubly robust unpenalized estimator. Simulations based on a similar data-generating process are investigated in more depth in Section~\ref{sec:simulation}.

\begin{figure}[ht]
    \centering
    \includegraphics[width=1\linewidth]{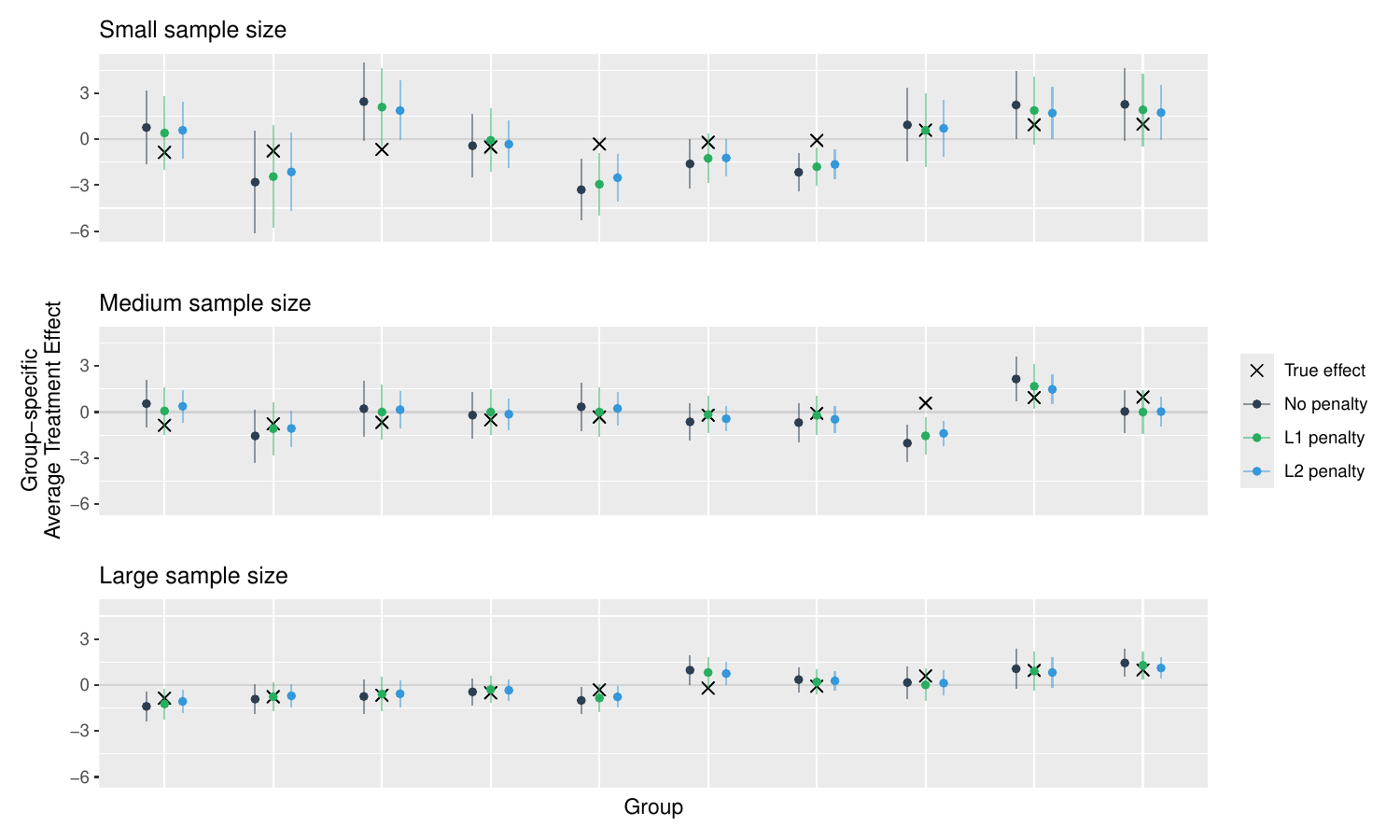}
    \caption{Example based on simulated data illustrating our penalized estimators applied to estimating group-specific average treatment effects. The true group-specific average treatment effects in each group are shown by the crosses. The estimated effects (point estimates and 95\% confidence intervals) for each group are shown based on a double-robust targeted estimator (first point in each group), an $L_1$ penalized estimator (second point), and an $L_2$ penalized estimator (third point) for small, medium, and large sample sizes. }
    \label{fig:group-effects-example}
\end{figure}

\paragraph{Prior Work} The utility of shrinkage estimators that trade bias for variance is well-known through the famous example of the James-Stein estimator, which demonstrates that in a certain normal mean model an estimator that scales unbiased initial estimates towards zero dominates the unbiased estimator in terms of joint MSE \citep{stein1956estimator, efron1977stein}. Similar James-Stein inspired estimators have also been derived in other contexts, such as simultaneous equations and two-stage least squares \citep{maasoumi1978stein, hansen20172sls}. The original James-Stein estimator can also be motivated by Empirical Bayes arguments \citep{efron2024empiricalbayes}. Indeed, shrinkage estimators are a major topic in Empirical Bayes methodology \citep{armstrong2022empiricalbayes}; we draw on such arguments to justify a simple modification of our $L_2$ penalization method to allow for adaptive shrinkage that depends on the precision of the individual parameter estimates. Our overall project is distinct from Empirical Bayes methods, however, in that we define our parameters via penalization. 

In another context, estimating regression coefficients with penalization in linear models was popularized by various regularized regression methods including the Lasso, Ridge, and Elastic-Net, to name only several examples in a vast literature \citep{tibshirani1996lasso,hoerl1970ridge,hui2005elasticnet}. These regression penalization methods yield estimators that trade bias for reduced variance, with a focus on improving predictive performance. Depending on the penalization term, the estimates of the regression coefficients can also be sparse, as is the case for the Lasso (using $L_1$ penalization). Our work diverges from this literature in that we make no modeling assumptions, and rather work within a fully non-parametric framework. In addition, our focus is on inference for statistical functionals, rather than on predictions, and our asymptotic results lead to straightforward constructions of confidence intervals. On the other hand, statistical inference for penalized regression coefficients typically depends on post-selection inference techniques \citep{lee2016postselection}. 

In applied Bayesian methodology, shrinkage of parameter estimates is ubiquitous through the application of priors. Bayesian shrinkage methods are appealing in that inference is available automatically via standard Bayesian arguments; for example, a common approach is to shrink parameter estimates in linear mixed models by placing hierarchical priors on the model coefficients. For treatment effect estimates in particular, \cite{feller2015hierarchical} advocate for shrinking multiple effect estimates (such as group-specific effects) towards a common mean and propose a parametric Bayesian modeling approach to that end. Our work has a similar goal, although we approach the problem in a frequentist non-parametric framework. 

In the context of causal inference, penalization has been previously investigated for estimating nuisance parameters that are involved in forming the final estimates of the causal target parameters of interest \citep{smucler2019l1, shortreed2017outcomelasso, benkeser2016hal}. However, achieving a desired bias-variance trade-off for the nuisance parameters does not necessarily imply that the subsequent estimates of the causal effects will share the same desirable properties. For example, using sparse regression methods for the nuisance parameters will not necessarily imply that the causal effect estimates are sparse. Our work takes a different approach by defining a new target parameter that incorporates the penalization. Nuisance parameters can be estimated using diverse methods, and are not limited to regularized regression methods, for example. 

Our work can be seen as a specific form of non-parametric Marginal Structural Model (MSM) proposed in the context of causal inference. Non-parametric MSMs summarize a possibly high-dimensional set of target parameters by projecting them onto a lower-dimensional working model. Such approaches have also been referred to as projection learners \citep{mcclean2024incremental}. Semi-parametric theory for a general class of MSMs is reviewed in \citet{susmann2023msm}. Also closely related to our work is that of \cite{bahamyirou2022drlasso}, who developed a penalized method for discovery of conditional average treatment effect (CATE) modifiers. The principal differences in our approaches lie in that ours is fully general and applicable to a large class of parameters beyond the CATE, and our results eschew any modeling assumptions such as the linear marginal structural model that their work imposes. 

\paragraph{Outline} The rest of the manuscript is organized as follows. In Section~\ref{sec:framework} we introduce a general class of penalized parameters. In Section~\ref{sec:theory} we derive the semi-parametric efficiency properties of this general parameter class. In Sections~\ref{sec:ridge} and \ref{sec:l1-penalty} we apply the results to parameters defined with $L_2$ and $L_1$ penalties, respectively. Simulation studies are included in Section~\ref{sec:simulation} and an application to estimating the performance of kidney dialysis providers is presented in Section~\ref{sec:application}. We conclude in Section~\ref{sec:discussion} with a discussion.

\section{Framework for Penalized Parameters}
\label{sec:framework}
Suppose we observe $n$ i.i.d. draws $O_1, \dots, O_n$ of the generic variable $O \in \mathcal{O}$ from a law $P_0$. We assume only that $P_0$ falls in the non-parametric model $\model$ (that is, $\model$ is the set of all probability laws defined on the support of $O$). Let $\param : \mathcal{M} \to \mathbb{R}^{\paramdim}$ be a vector-valued parameter defined by $\param(P) = ( \param_d(P) : d \in \mathcal{D} )$, where $\param_d : \model \to \mathbb{R}$ is a statistical functional indexed by $d \in \mathcal{D}$. We assume throughout that the $\param_d$ are sufficiently smooth so as to be \textit{pathwise differentiable}, a concept introduced in the next section. 

\paragraph{Notation} Whenever there is a set or vector $\mathcal{D}$, we will use the subscript `$d$' to denote the $d$th element of the set, as in $\param_d$ for the $d$th element of the vector $\param(P)$. When we make a statement concerning ``the $\param_d$'' we are applying the statement to all $\param_d$ for $d \in \mathcal{D}$. For convenience we will use the subscript `0' to denote a parameter evaluated at $P_0$, e.g., $\param_0 = \param(P_0)$. We will also use the subscript `$n$' to signal dependence on $n$; for example, we will write $\param_n$ to denote an estimator of $\param_0$.
For a function $f$ and $P \in \model$ we write the expectation of $f$ with respect to $P$ as either $\E_P[f]$ or $Pf = \int f dP$. We may write expectation with respect to the empirical measure as $P_n f = n^{-1} \sum_{i=1}^n f(O_i)$.  A reference table listing key notation is provided in Appendix~\ref{appendix:notation}.

\paragraph{Examples} Now we introduce three examples of vector-valued statistical parameters. We will use these parameters later to evaluate our proposed methods in simulation studies.

\begin{example}[Non-parametric linear association]
    \label{example:non-parametric-coefficient}
    Let $O = (X, Y)$ where $X = (X_1, \dots, X_D)$ is a $D$-dimensional vector of covariates and $Y \in \mathbb{R}$ is a continuous outcome. Denote by $X_{(-d)} = (X_1, \dots, X_{d-1}, X_{d+1}, \dots, X_D)$ the vector containing all but the $d$th element of $X$. For each $d \in \mathcal{D} = \{1, \dots, D\}$, define
    \begin{align}
        \param_d(P) = \E_P\left[\Cov_P\left(Y, X_d | X_{(-d)}\right)\right].
    \end{align}
    Collecting these into a vector yields the parameter $\param(P) = \{ \param_d(P) : d \in \mathcal{D} \}$. 


    Note that this parameter has the useful property that it can be estimated using linear regression, in that in a main-terms linear regression of $Y$ on $X$ the coefficient estimate $\hat{\beta}_d$ converges to $\param_d(P) \slash \E_P\left[\Var_P\left(Y | X_d\right) \right]$. This property will allow us to compare our methods directly to penalized generalized linear models in simulations.
\end{example}


\begin{example}[Group-specific treatment effects]
    Let $X$ be a vector of covariates, $G \in \{1, \dots, D \}$ a variable indexing assignment to a group, and $A \in \{0, 1 \}$ a binary treatment. Let $Y(0), Y(1) \in \mathbb{R}$ be potential outcomes corresponding to treatment assignments $A = 0$ and $A = 1$, respectively, and let $Y = AY(1) + (1 - A)Y(0)$ be the observed outcome. The observed data are therefore $O = (X, G, A, Y)$. 
    The causal parameter of interest is the group-specific average treatment effect, denoted in terms of potential outcomes as, for $d \in \mathcal{D} = \{1, \dots, D \}$,
    \begin{align}
        \param_d^*(P) = \E_P\left[ Y(1) - Y(0) | G = d \right]. 
    \end{align}
    Let $\outcomemodel_P(a, d, X) = \E_P[Y \mid A = a, G = d, X]$. Then, under standard causal assumptions (conditional ignorability and positivity), the parameter $\param_d(P)$ is identified in terms of only observable variables as, for $d \in \mathcal{D}$,
    \begin{align}
        \param_d(P) = \E_P\left[ \outcomemodel_P(1, d, X) - \outcomemodel_P(0, d, X) \mid G = d \right]. 
    \end{align}
\end{example}

\begin{example}[Indirectly standardized outcomes]
    Let $X$ be a vector of covariates and $A \in \mathcal{D} = \{ 1, \dots, D \}$ a categorical treatment indicator. Let $\{Y(a) : a \in \mathcal{D}\}$ be a set of potential outcomes corresponding to each of the treatment assignments, and $Y = Y(A)$ be the outcome under the observed treatment assignment. The observed data comprise $O = (X, A, Y)$. 

    Let $Z \sim P_{A|X}$ be a random draw from the conditional distribution of the treatment assignment given covariates. Let $Y(Z)$ be the potential outcome under the stochastic intervention in which the individual was reassigned to treatment $Z$ (which possibly differs from the observed treatment assignment $A$). The target causal parameter is defined as:
    \begin{align}
        \param^*_d(P) = \E_P[Y(Z) | A = d].
    \end{align}
    That is, $\param^*_d(P)$ is the expected outcome if, possibly contrary to fact, all observations from group $d$ were randomly reassigned to an alternative treatment $Z$. 

    Let $\outcomemodel_P(X) = \E_P[Y \mid X]$. Then $\param^*(P)$ is identified using only observable variables as
    \begin{align}
        \param_d(P) = \E_P[\outcomemodel(X) | A = d]. 
    \end{align}
    
    The parameter $\param_d$ is sometimes referred to as an \textit{indirectly standardized outcome}. One application is in provider profiling, where the observations are patients with baseline characteristics $X$ who were treated at healthcare provider $A$ and experienced the outcome $Y$. One way of evaluating the performance of a provider is to ask what would have happened if the population of patients who were treated by that provider had instead been randomly reassigned for treatment to another provider that tends to treat a similar patient population. This counterfactual parameter can be estimated by comparing $\param_d(P)$ to the mean outcome of patients treated at the provider; for example, through the difference $\param_d(P) - \E_P[Y \mid A = d]$ \citep{daignault2016indirect, diaz2023nonagency, susmann2024doublyrobustnonparametricefficient}. 
\end{example}

\paragraph{Penalized Parameter} We now define a novel \textit{penalized  parameter} defined in terms of the original parameter $\param$. For any $P \in \model$, define the penalized parameter $\penalparam_\lambda \in \mathbb{R}^{\paramdim}$ as the solution to the following optimization problem:
\begin{align}
    \label{eq:penalized-parameter-def}
  \penalparam_\lambda(P) = \argmin_{\penalparam \in \mathbb{R}^{\paramdim}} \objective_\lambda(\param(P), \penalparam),
\end{align}
where the optimization objective $\objective_\lambda : \mathbb{R}^{\paramdim} \times \mathbb{R}^{\paramdim} \to \mathbb{R}$ is the map
\begin{align}
    (x, \tilde{x}) \mapsto \objective_\lambda(x, \tilde{x}) =  \loss(x, \tilde{x}) + \penaltyf_\lambda(\tilde{x}).
\end{align}
The loss function $L : \mathbb{R}^{\paramdim} \times \mathbb{R}^{\paramdim} \to \mathbb{R}$ measures the fidelity of the penalized parameter to the original parameter, and $\penaltyf_\lambda : \mathbb{R}^{\paramdim} \to \mathbb{R}$ is a penalization term. The tuning parameter $\lambda \in \Lambda$ controls the strength of the penalization. Typically, $\Lambda = \mathbb{R}_{>0}$; that is, $\lambda$ is a positive number, with higher values of $\lambda$ implying stronger penalization.
Further assumptions may be necessary to assure that the optimization problem in \eqref{eq:penalized-parameter-def} has a unique solution and that subsequently $\penalparam(P)$ is well-defined. We first consider the case where the penalization parameter $\lambda$ is fixed and user-defined. After developing theory for the case of fixed $\lambda$, we apply the results to suggest optimal data-adaptive methods for choosing $\lambda$.  

\section{General results}
\label{sec:theory}
In this section we review foundations from semi-parametric efficiency theory, which we then apply to derive the semi-parametric efficiency bound for estimating the penalized parameter $\penalparam_{\lambda}$ at the true data-generating distribution $P_0$ under sufficiently smooth choices of loss function and penalty term. A general estimator based on one-step estimation that achieves the efficiency bound is presented in Appendix~\ref{appendix:one-step}. Accessible and high-quality reviews of the relevant semi-parametric theory, with an emphasis on applications to causal inference, can be found in \cite{kennedy2016theory, kennedy2024review}. Other key references include \cite{vanderVaart&Wellner96, vanderVaart98, Bickel97}. 

\subsection{Semi-parametric efficiency theory and one-step estimation}
For the purposes of introducing the principal concepts, consider a generic statistical functional $\phi : \model \to \mathbb{R}^p$ (for $p \geq 1$). We focus on functionals that are sufficiently smooth so as to be \textit{pathwise differentiable}, as this is a crucial property that allows for the derivation of non-parametric efficiency bounds. To introduce pathwise differentiability, for every $P \in \model$ and $s \in L_0^2(P)$, $s$ bounded and not identically zero, define a parametric submodel $\mathcal{P}_s = \{ P_{s, \epsilon} : \epsilon \in \mathbb{R}^p, \|\epsilon\|_\infty < \| s\|_\infty^{-1} \} \subset \model$, where $dP_{s,\epsilon} = (1 + \epsilon^\top s) dP$. Note that $\mathcal{P}_s$ is a fluctuation of $P$ in the direction $s$, in the sense that $P_{s,\epsilon} = P$ at $\epsilon = 0$ and the score of $P_{\epsilon, s}$ at $\epsilon = 0$ is $s$. We call $\phi$ pathwise differentiable at $P$ if there exists a functional $\eif_\phi(P) : \mathcal{O} \to \mathbb{R}^p$ with mean zero and finite variance referred to as an \textit{influence curve} such that, for every $s$, the following derivative exists and can be expressed as
\begin{align}
    \frac{\partial}{\partial \epsilon} \phi(P_{s, \epsilon}) \Bigg|_{\epsilon = 0} = \E_P\left[ s(O) \eif_\phi(P)(O) ^\top \right].
\end{align}
Because every $s \in L_0^2(P)$ induces a fluctuation model $\mathcal{P}_s$, if the derivative exists then  $\eif_\phi(P)$ is unique and is referred to as the \textit{efficient influence function} of $\phi$ at $P$. A key result of semi-parametric efficiency theory is that the asymptotic covariance of any regular estimator of $\phi(P)$ is lower bounded by the variance of the efficient influence function:
\begin{align}
    \sigma^2_\phi(P) = \E_P[\eif_\phi(P)(O) \eif_\phi(P)(O)^\top ]. 
\end{align}
When a parameter is pathwise differentiable, the influence curve serves as the first-order term of a type of distributional Taylor expansion of the parameter. Formally, for any $P_1, P_2 \in \mathcal{M}$, write
\begin{align}
    \label{eq:von-mises}
    \phi(P_1) - \phi(P_2) = (P_1 - P_2)\influencef_\phi(P_1) + R(P_1, P_2),
\end{align}
for an influence function $\influencef_\phi(P_1) : \mathcal{O} \to \mathbb{R}^p$ of $\phi$ at $P$ and second-order remainder term $R : \model \times \model \to \mathbb{R}^p$. The remainder term is called second-order because $R$ is is a function only of squares or products of differences in its arguments. This expansion is sometimes referred to as the \textit{von-Mises} expansion of the parameter \citep{mises1947asymptotic}. 

Our analyses of the semi-parametric efficiency properties of the proposed penalized parameters therefore proceeds in two steps: first, we establish whether the parameter is pathwise differentiable, and, if so, derive the form of its efficient influence function and the associated second-order remainder term. By characterizing the form the EIF and the remainder term we can propose estimators, and subsequently establish conditions under which that estimator is consistent, efficient, and asymptotically normal. 

In this work we focus on penalized parameters defined with respect to an underling parameter $\param$ that is pathwise differentiable and admits a von-Mises expansion of the form \eqref{eq:von-mises}. For the three example target parameters we give below the form of their associated efficient influence functions and the remainder term of the von-Mises expansion.
\setcounter{example}{0}
\begin{example}[Non-parametric regression coefficient (continued)]
 Let $\propscore_P(X_{(-d)}) = \E_P[X_d \mid X_{(-d)}]$ and $\outcomemodel_P(X_{(-d)}) = \E_P[Y \mid X_{(-d)}]$. The parameter $\param_d$ is pathwise differentiable with efficient influence function $\eif_{\param_d}$ at $P$ characterized by
    \begin{align}
        \eif_{\param_d}(P)(O) = \left\{ X_d - \propscore_P(X_{(-d)}) \right\} \left\{ Y - \outcomemodel(X_{(-d)})] \right\}.
    \end{align}
    Furthermore, $\param_d$ satisfies a von-Mises expansion with remainder term $R_d$ for any $P, P_0 \in \model$ characterized by
    \begin{align}
        R_d(P_0, P) = \E_{P_0} \left[ \left\{ \pi_P\left(X_{(-d)}\right) - \pi_0\left(X_{(-d)}\right) \right\} \left\{ \mu_P\left(X_{(-d)}\right) - \mu_0\left(X_{(-d)}\right) \right\} \right].
    \end{align}
\end{example}

\begin{example}[Group-specific treatment effects (continued)]
    Fix $d \in \mathcal{D}$. Let $\propscore_P(d, a, X) = P(A = a \mid G = d, X)$ and $\outcomemodel_P(d, a, X) = \E_P[Y \mid A = a, G = d, X]$. The parameter $\param_d$ is pathwise differentiable with efficient influence function $\eif_{\param_d}$ at any $P \in \model$ characterized by 
    \begin{align}
        \eif_{\param_d}(P)(O) = \frac{\I[G = d]}{P(G = d)} \left[ \frac{2A - 1}{\propscore_P(d, A, Y)} \left( Y - \outcomemodel_P(G, A, X) \right) + \outcomemodel(d, 1, X) - \outcomemodel(d, 0, X) - \param_d(P) \right]. 
    \end{align}
    The parameter $\param_d$ satisfies a von-Mises expansion with remainder term $R_d$ for any $P, P_0 \in \model$ characterized by
    \begin{align}
        & R_d(P_0, P) \\
        & = \sum_{a \in \{0, 1\}} \frac{2a - 1}{P(G = d)} \E_{P_0}\left[ \I[A = d] \left\{ \frac{1}{\propscore_P(d, a, X)} - \frac{1}{\propscore_0(d, a, X)} \right\} \left\{ \outcomemodel_0(d, a, X) - \outcomemodel_P(d, a, X) \right\} \propscore_0(d, a, X)   \right]. 
    \end{align}
\end{example}


\begin{example}[Indirectly standardized outcomes ]
    Fix $d \in \mathcal{D}$. Let $\propscore_P(a, X) = P(A = a \mid X)$ and $\outcomemodel_P(X) = \E_P[Y \mid X]$. The indirectly standardized outcome parameter $\param_d$ is pathwise differentiable \citep{susmann2024doublyrobustnonparametricefficient} with efficient influence function $\eif_{\param_d}$ at any $P \in \model$ characterized by
    \begin{align}
        \eif_{\param_d}(P)(O) = \frac{1}{P(A = d)} \left\{ \propscore_P(d, X) \left(Y - \outcomemodel_P(X)\right) + \I[A = d] \left( \outcomemodel_P(X) - \param_d(P) \right) \right\}. 
    \end{align}
    The parameter $\param_d$ satisfies a von-Mises expansion with remainder term $R$  for any $P, P_0 \in \mathcal{M}$ characterized by
    \begin{align}
        R_d(P_0, P) = \E_{P_0}\left[ \frac{1}{P(A = d)} \left( \propscore_P(d, X) - \propscore_0(d, X) \right) \left(  \outcomemodel_0(X) - \outcomemodel_P(X) \right) \right]. 
    \end{align}
\end{example}

\subsection{Pathwise differentiability of general penalized parameters}
\label{sec:generalized-parameter-results}
In the following theorem, we provide conditions under which $\penalparam_\lambda$ is pathwise differentiable and provide the form of its EIF when the penalization tuning parameter $\lambda$ is fixed. Theory for the fixed $\lambda$ scenario is useful for two reasons. First, doing so leads to strategies for choosing $\lambda$ data-adaptively. Second, as we show in the next section, when $\lambda$ is itself estimated from the data and applied to form a penalized parameter $\penalparam_{\lambda}$, the uncertainty arising from estimating $\lambda$ is asymptotically negligible; in other words, under mild conditions the estimated $\lambda$ can be treated as fixed, and the results proved here for fixed $\lambda$ can be applied.

The following theorem and its conditions are an adaption of \citealt[Theorem 1]{susmann2023msm}. The proof is a straightforward application of the proof of that theorem, and is therefore omitted.
\begin{theorem}[Efficient influence function of $\penalparam_\lambda$ for fixed $\lambda$]
    \label{thm:penalized-eif}
    Fix $\lambda \in \Lambda$. Assumptions:
    \begin{enumerate}
        \item \label{assumption:eif-param-pathwise-differentiable} The parameter $\param$ is pathwise differentiable at any $P \in \model$ with EIF $\eif_{\param}(P) : \mathcal{O} \to \mathbb{R}^{\paramdim}$. 
        \item \label{assumption:eif-derivatives} For every $x \in \mathbb{R}^{|\mathcal{D}|}$, the following conditions are met:
        \begin{enumerate}
            \item $\tilde{x} \mapsto U_\lambda(x, \tilde{x})$ is differentiable at every $\tilde{x}$ with derivative $\dot{U}_\lambda(x, \tilde{x}) \in \mathbb{R}^{|\mathcal{D}|}$. 
            \item $\tilde{x} \mapsto \dot{U}_\lambda(x, \tilde{x})$ is differentiable at every $\tilde{x}$ with derivative $\ddot{U}_\lambda(x, \tilde{x}) \in \mathbb{R}^{|\mathcal{D}| \times |\mathcal{D}|}$. 
        \end{enumerate}
        In addition, for every $\tilde{x} \in \mathbb{R}^{|\mathcal{D}|}$, it holds that
        \begin{enumerate}
            \item $x \mapsto \dot{U}_\lambda(x, \tilde{x})$ is differentiable at every $x \in \mathbb{R}^{|\mathcal{D}|}$ with derivative $\nabla \dot{U}_\lambda(x, \tilde{x}) \in \mathbb{R}^{|\mathcal{D}|}$, and $\nabla \dot{U}_\lambda(x, \tilde{x})$ is invertible. 
        \end{enumerate}
    \end{enumerate}
    Then the functional $P \mapsto \penalparam_\lambda(P)$ is pathwise differentiable at every $P \in \model$ with an efficient influence function $\eif_{\penalparam_\lambda}(P)$ at $P$ given by
    \begin{align}
        O \mapsto \eif_{\penalparam_\lambda}(P)(O) = M^{-1} \left[ \nabla \dot{\objective}_\lambda \left(\param(P), \penalparam(P)\right) \times \eif_\param(P)(O) + \dot{\objective}_\lambda\left(\param(P), \penalparam(P) \right) \right],
    \end{align}
    where the normalizing matrix $M$ is given by
    \begin{align}
        M = -\ddot{\objective}_\lambda (\param(P), \penalparam(P)).
    \end{align}
\end{theorem}
The required pathwise differentiability of $\param$ (Assumption~\ref{assumption:eif-param-pathwise-differentiable}) must be verified separately for the specific choice of underlying parameter, as we have done for the three examples. Assumption~\ref{assumption:eif-derivatives}, requiring that various derivatives of objective function exist, must be verified for each choice of loss function and penalty term. 

In Appendix~\ref{appendix:one-step} we describe a one-step estimator for $\penalparam_\lambda$ when $\lambda$ is fixed that, under mild conditions, is consistent, asymptotically normal, and achieves the non-parametric efficiency bound implied by the form of the EIF given in Theorem~\ref{thm:penalized-eif}. In the following we focus on the scenario in which the tuning parameter is chosen data-adaptively.

%
%

\section{\texorpdfstring{$L_2$}{L2} penalty}
\label{sec:ridge}
In many real-world scenarios we wish to choose the penalization tuning parameter data-adaptively in order to yield an estimator with desirable finite-sample properties. In this section we consider the choice of tuning parameter when using the $L_2$-norm penalty. We start with the $L_2$-norm because its infinite differentiability leads to particularly tidy results. Throughout, we use a squared-error loss function $\loss(x, \tilde{x}) = \| x - \tilde{x}\|_2^2$. For the penalty term, let $\penaltyf_2(\tilde{x}) = \lambda \| \tilde{x} \|^2_2$. Begin by fixing a $\lambda \geq 0$. The objective function is then
\begin{align}
    \objective_\lambda(x, \tilde{x}) = \| x - \tilde{x} \|_2^2 + \lambda \|\tilde{x}\|_2^2,
\end{align}
and the optimization problem \eqref{eq:penalized-parameter-def} has the solution, for any $P \in \model$,
\begin{align}
    \penalparam_\lambda(P) = \frac{1}{1+\lambda} \param(P).
\end{align}
Applying Theorem~\ref{thm:penalized-eif} (Assumption~\ref{assumption:eif-derivatives} thereof easily verified due to the infinite differentiability of the $L_2$-norm) shows that the EIF of $\penalparam_\lambda$ is simply the scaled EIF of $\param$:
\begin{align}
    \label{eq:ridge-eif-fixed-lambda}
    \eif_{\penalparam_\lambda}(P)(O) = \frac{1}{1+\lambda} \eif_\param(P)(O).
\end{align}
Indeed, the machinery of Theorem~\ref{thm:penalized-eif} isn't necessary to derive the above EIF, as it follows straightforwardly from the fact that $\penalparam_\lambda$ is simply a scaled version of $\lambda$. 

In practice, we often do not have a value of $\lambda$ fixed a priori; rather, we wish to choose $\lambda$ data-adaptively. We propose choosing $\lambda$ by minimizing the following criterion as a function of $\lambda$, which we denote $\criterion$:
\begin{align}
 \criterion(\lambda, \param(P), \sigma^2_\param(P), n) = \frac{\lambda^2}{(1+\lambda)^2} \| \param(P) \|_2^2 + \frac{1}{n (1+\lambda)^2} \mathrm{tr}\left( \sigma^2_{\param}(P) \right), 
\end{align}
where $n \geq 0$. The data-adaptive choice of $\lambda$ is then given by
\begin{align}
    \lambda^* = \argmin_{\lambda \geq 0} \criterion(\lambda, \param(P), \sigma^2_\param(P), n). 
\end{align}
We argue that this is a useful way to choose $\lambda$ because the criterion can be understood as an asymptotically valid approximation of the mean-squared error of an estimator of $\penalparam_\lambda$ relative to the true parameter value $\param$. To illustrate this, for any $P \in \model$ define for the mean squared error (MSE) of an estimator $\penalparam_{\lambda,n}$ of $\penalparam_\lambda(P)$ relative to  $\param(P)$ as
\begin{align}
    \label{eq:mse}
    \MSE(\penalparam_{\lambda, n}, \param(P)) = \Bias(\penalparam_{\lambda, n}, \param(P))^2 + \Variance(\penalparam_{\lambda, n})
\end{align}
where $\Bias(\penalparam_{\lambda, n}, \param(P))^2 = \| \E_P[\penalparam_{\lambda, n}] - \param(P)\|^2_2$, $\Variance(\penalparam_{\lambda, n}) = \mathrm{tr}\left( \mathsf{Var} \left[ \penalparam_{\lambda, n} \right] \right)$, and $\mathrm{tr}$ is the matrix trace operator. An asymptotically normal and efficient estimator $\penalparam_{\lambda, n}$ of $\penalparam(P)$ satisfies
\begin{align}
    \sqrt{n}(\penalparam_{\lambda, n} - \penalparam_{\lambda}(P)) \convd N\left(0, \sigma^2_{\penalparam_\lambda}(P) \right).
\end{align}
Therefore, an asymptotically valid estimate of the variance of $\penalparam_{\lambda, n}$ is $\sigma^2_{\penalparam_\lambda}(P) / n$. Using this as an estimate of the variance yields a simple form for the MSE \eqref{eq:mse}:
\begin{align}
    \label{eq:ridge-mse}
    & \frac{\lambda^2}{(1+\lambda)^2} \| \param(P) \|_2^2 + \frac{1}{n (1+\lambda)^2} \mathrm{tr}\left( \sigma^2_{\penalparam_{\lambda}}(P) \right) = \criterion(\lambda, \param(P), \sigma^2_\param(P), n). 
\end{align}
Therefore, minimizing $\criterion$ as a function of $\lambda$ can be seen as minimizing an asymptotic approximation of the MSE of the penalized estimator relative to the true parameter. The major caveat with this choice is that it depends on an asymptotic approximation of the variance of the estimator. If finite-sample expressions of the bias and variance of the estimator are available, then they could be used as a more accurate alternative.

Conveniently, there is a closed-form solution for the value of $\lambda$ that minimizes $\criterion$. To express the closed form solution succinctly, first define, for any $P \in \mathcal{M}$ such that $\| \param(P) \|_2^2 > 0$, the parameter $\gamma : \model \to \mathbb{R}$ as 
\begin{align}
    P \mapsto \gamma(P) = \frac{\mathrm{tr}(\sigma^2_{\param}(P))}{\| \param(P) \|_2^2}. 
\end{align}
The parameter $\gamma$ is interesting in its own right as a summary of the efficiency bound of $\param$ relative to the overall scale of $\param$, and its squared root is often referred to as the coefficient of variation. In addition, it is useful because the value of $\lambda$ that minimizes the MSE given in \eqref{eq:ridge-mse} is a simple function of $\gamma(P)$:
\begin{align}
    \label{eq:optimal-ridge-mse}
    \lambda^*(\gamma(P), n) = \frac{1}{n} \times \gamma(P).
\end{align}
For intuition, $\lambda^*(\gamma(P), n)$ has a simple interpretation as the ratio of the sum of the (approximate) variance of the estimator of each parameter divided by the square of each parameter. Thus, when the variance is low relative to the magnitude of the parameter, less shrinkage is applied, and vice versa when the variance is high. 

We continue by studying the semi-parametric efficiency properties of the parameter $\gamma$.
Because $\gamma$ is a differentiable function of $\param$ and $\sigma^2_\param$, it follows that it will be pathwise differentiable so long as the same holds for $\param$ and $\sigma^2_\param$.
The following theorem formalizes this result.
\begin{lemma}[Efficient influence function of $\gamma$]
    \label{lemma:gamma-eif}
    For all $d = 1, \dots, D$, assume that $\sigma^2_{\param_d}$ is pathwise differentiable at any $P \in \model$ with EIF $\eif_{\sigma^2_{\param_d}}(P) : \mathcal{O} \to \mathbb{R}$. 
    Then the parameter $\gamma$ is pathwise differentiable with EIF $\eif_{\gamma}(P) : \mathcal{O} \to \mathbb{R}$ at $P \in \model$ characterized by 
    \begin{align}
        O \mapsto \eif_{\gamma}(P)(O) =  -2 \times \frac{\mathrm{tr}(\sigma^2_\phi(P))}{\| \param(P) \|_2^3 } \sum_{d=1}^D \eif_{\param, d}(P)(O) +\frac{\sum_{d=1}^D \eif_{\sigma_d^2}(P)(O)}{\| \param(P) \|_2^2 }.
    \end{align}
\end{lemma}

We can now go one step further and derive the EIF of the penalized parameter $\penalparam_{\lambda^*}$, the penalized parameter where the optimizer $\lambda^*(\gamma(P), n)$ is chosen as the penalization parameter. 
\begin{theorem}[Efficient influence function of $\penalparam_{\lambda^*}$]
\label{thm:lambda-star-eif}
    Fix $n > 0$. For any $P \in \model$, set $\lambda^* = \frac{1}{n} \gamma(P)$. The parameter $\penalparam_{\lambda^*}$ is pathwise differentiable at $P$ with EIF $\eif_{\penalparam_{\lambda^*}}(P) : \mathcal{O} \to \mathbb{R}^{\paramdim}$ characterized by
    \begin{align}
        \eif_{\penalparam_{\lambda^*}}(P)(O) = \frac{1}{1 + \lambda^*} \eif_{\param}(P)(O) - \frac{1}{n} \times \frac{\param(P)}{(1 + \lambda^*)^2} \eif_{\gamma}(P)(O).
    \end{align}
\end{theorem}
The first term of the EIF for $\penalparam_{\lambda^*}$ is simply the EIF of the original parameter scaled by $\lambda^*$; this term can be interpreted as representing uncertainty in estimating $\penalparam_{\lambda^*}$ when $\lambda^*$ as fixed, as in \eqref{eq:ridge-eif-fixed-lambda}. The second term represents uncertainty in estimating $\lambda^*$. Notably, this term is scaled by $1/n$. This suggests that the second term of the EIF will be negligible as $n$ increases. 

An estimator of $\penalparam_{\lambda^*}$ could be constructed  using the full EIF of $\penalparam_{\lambda^*}$ given in Theorem~\ref{thm:lambda-star-eif} (for example, using the one-step approach described in Appendix~\ref{appendix:one-step}). However, doing so would require estimating $\eif_{\gamma}$, which may be difficult or involve estimating additional nuisance parameters beyond those required for estimating $\param$, $\eif_\param$ and $\lambda^*$. Therefore, we propose forming a simpler estimator that disregards the $\eif_{\gamma}$ term. We subsequently prove that ignoring this term is justified in an asymptotic analysis.

To form the estimator, suppose that we have an asymptotically normal and efficient estimator $\param_n$ of $\param_0$, and a consistent estimator $\gamma_n$ of $\gamma_0$. We propose setting the penalty term to $\lambda_n^* = \frac{1}{n}\gamma_n$ and estimating $\penalparam_{\lambda_n^*}$ by simply scaling $\param_n$ by the estimated shrinkage factor:
\begin{align}
    \label{eq:simple-l2-estimator}
    \penalparam_{\lambda_n^*, n} = \frac{1}{1 + \lambda_n^*} \param_n. 
\end{align}
To justify this simplified estimator, we prove the following alternative decomposition of the penalized parameter that shows, if the original parameter admits a von-Mises expansion, then the penalized parameter satisfies a similar expansion that differs only by terms related to $\lambda^*$. 
The proof is provided in Appendix~\ref{appendix:l2-penalized-approximation}.
\begin{theorem}
    \label{thm:alternative-decomposition}
    Suppose that $\param$ satisfies a von-Mises expansion of the form \eqref{eq:von-mises} with EIF $\eif_\param$ and second-order remainder $R$. Fix $n > 0$ and let $\lambda^* = \frac{1}{n}\gamma(P)$. Let $\penalparam_{\lambda^*} = \frac{1}{1+\lambda^*} \param(P)$, and assume that $\penalparam_\lambda$ satisfies a von-Mises expansion with EIF $\eif_{\penalparam}$ and second-order remainder $R_{\penalparam}$. 
    Then the parameter $\penalparam_{\lambda^*}$ satisfies the following expansion:
    \begin{align}
        \penalparam_{\lambda^*}(P_1) - \penalparam_{\lambda^*}(P_2) =& -P_2\left[ \frac{1}{1 + \lambda^*(P_1)}  \eif_\param(P_1) \right] + \left\{ \frac{1}{1 + \lambda^*(P_1)} - \frac{1}{1 + \lambda^*(P_2)} \right\} \param(P_2) \\
        &+ \frac{1}{1+\lambda^*(P_1)} R(P_1, P_2). 
    \end{align}
\end{theorem}
This result is notable because as $n \to \infty$ the decomposition converges to
\begin{align}
    \penalparam_{\lambda^*}(P_1) - \penalparam_{\lambda^*}(P_2) = -P_2[\eif_{\param}(P_1)] + R(P_1, P_2). 
\end{align}
The proof is given in Appendix~\ref{appendix:l2-expansion-proof}.
Asymptotic consistency, normality and efficiency therefore follows for $\penalparam_{\lambda_n^*}$ under the same conditions necessary for the original parameter $\param$, with the only other condition necessary being that an estimator $\gamma_n$ of $\gamma_0$ does not diverge This is formalized in the following theorem, which establishes conditions under which $\penalparam_{\lambda^*}$ is asymptotically normal and efficient estimator of $\param_0$. 
\begin{theorem}[Asymptotic normality and efficiency of $\penalparam_{\lambda_n^*}$ for $L_2$-penalization]
    \label{thm:asymptotic-normality-estimated-lambda}
    Let $\param_n$ and $\gamma_n$ be estimators of $\param_0$ and $\gamma_n$, respectively. Let $\lambda_n^* = \frac{1}{n} \times \gamma_n$.  Assume each of the following:
    \begin{enumerate}
        \item \label{assumption:l2-efficiency} The estimator $\param_n$ is asymptotically normal and efficient:
        \begin{align}
            \sqrt{n}\left( \param_n - \param_0 \right) \convd N\left(0, \sigma_{\param,0}^2 \right). 
        \end{align}
        \item \label{assumption:l2-gamma-consistency} The estimator $\gamma_n$ is converges: there exists a $\gamma_\infty$ with $\infty < \gamma_\infty < \infty$ such that $\gamma_n - \gamma_\infty = o_P(1)$. 
    \end{enumerate}
    Then $\penalparam_{\lambda_n^*, n} = \frac{1}{1 + \lambda_n^*} \param_n$ is an asymptotically normal and efficient estimator of $\param_0$:
    \begin{align}
        \sqrt{n}\left( \penalparam_{\lambda_n^*} - \param_0 \right) \convd N\left(0, \sigma^2_{\param, 0}\right).
    \end{align}
\end{theorem}
\begin{proof}
    By assumption, $\gamma_n - \gamma_\infty = o_P(1)$. Therefore $\lambda_n^* = o_P(1)$, and furthermore the shrinkage factor $1 \slash (1 + \lambda^*_n) = 1 + o_P(1)$. Thus, Slutsky's theorem and the fact that the estimator of $\param$ is asymptotically normal and efficient implies the stated result.
\end{proof}

Establishing conditions under which Assumption~\ref{assumption:l2-efficiency} holds depends on the underlying parameter of interest. Typically convergence of $\gamma_n$, required by Assumption~\ref{assumption:l2-gamma-consistency}, will hold under weak assumptions; indeed, $\gamma_n$ will typically be a consistent estimator of $\gamma_0$ under the same assumptions necessary for Assumption~\ref{assumption:l2-efficiency}. In the interest of generality, Theorem~\ref{thm:asymptotic-normality-estimated-lambda} is stated in terms of a generic asymptotically efficient estimator $\psi_n$ of $\psi_0$. Alternatively, one could use the expansion  in Theorem~\ref{thm:alternative-decomposition} to construct an estimator of $\penalparam_0$, e.g. by using a one-step estimation strategy. 

Based on the asymptotic normality result of Theorem~\ref{thm:asymptotic-normality-estimated-lambda}, a straightforward and asymptotically valid $(1 - \alpha)\times100\%$ confidence interval for $\param$ can be formed using the estimated variance of the unpenalized parameter estimate:
\begin{align}
    \label{eq:basic-confidence-intervals}
    C_{1-\alpha}(\penalparam_{\lambda_n^*}) = \left( \penalparam_{\lambda_n^*} - q_{1 - \alpha} \sqrt{\frac{\sigma^2_{d,n}}{n} }, \penalparam_{\lambda_n^*} + q_{1 - \alpha} \sqrt{\frac{\sigma^2_{d,n}}{n} } \right), 
\end{align}
where $\sigma^2_{d,n}$ is an estimate of the efficiency bound of $\param_d$. Assuming that we have access to an asymptotically normal and efficient estimator of $\param_d$, then such an estimate of the efficiency bound is typically available through the estimator's reported standard error. This is similar to the recent proposal in \cite{kaplan2024} for forming confidence intervals of biased parameters that are centered on the biased parameter estimate, but use the standard error of the original (unbiased) estimator to determine the confidence interval width. 

The above confidence interval is asymptotically valid, but not entirely satisfying as it has the same width as a confidence interval for the unpenalized parameter. As an alternative, we can form a narrower confidence interval based on the estimated shrinkage factor:
\begin{align}
    C'_{1-\alpha} = \left( \param_{n} -  \frac{q_{1 - \alpha}}{1 + \lambda_n^*} \sqrt{\frac{\sigma^2_{d,n}}{n} }, \param_{n} +  \frac{q_{1 - \alpha}}{1 + \lambda_n^*} \sqrt{ \frac{\sigma^2_{d,n}}{n} } \right). 
\end{align}
The asymptotic validity of the confidence interval follows from the same logic as the proof of Theorem~\ref{thm:asymptotic-normality-estimated-lambda}. 

In some applications, the fact that the penalized estimator $\penalparam_{\lambda_n^*}$ shrinks all estimates by the same factor $1 \slash (1 + \lambda_n^*)$ may not be desirable. Instead, we may wish to shrink each estimate in a manner proportional to the precision of the estimate. To propose such an estimator, note that we can rewrite the penalized parameter in the following form:
\begin{align}
    \label{eq:l2-form-variances}
    \penalparam_{\lambda_n^*}(P) &= \frac{1}{1 + \lambda_n^*(P) } \param(P) \\
    &= \frac{\frac{1}{D}\|\param(P)\|_2^2}{\frac{1}{D} \|\param(P)\|_2^2 + \frac{1}{D} \sum_{d'=1}^D \frac{1}{n} P\left[ \eif_{\param, {d'}}(P)^2 \right]} \param(P). 
\end{align}
In this form, the shrinkage is recognizable as the ratio involving the variance of the original parameter $\param$ around zero and the mean of the approximate estimator variances. This form also suggests a simple modification to allow for variable shrinkage. For a parameter $\param_d$ ($d \in \mathcal{D}$), estimate the shrinkage using the approximate estimator variance of only $\param_d$:
\begin{align}
    \penalparam_d^\mathsf{eb}(P) &= \frac{\frac{1}{D}\|\param(P)\|_2^2}{\frac{1}{D}\|\param(P)\|_2^2 + \frac{1}{n} P\left[ \eif_{\param, {d}}(P)^2 \right]} \param(P). 
\end{align}
This estimator has a natural connection to Empirical Bayes, as it can be interpreted as the posterior mean of $\param_d$ under a normal observation model with $\param_{d,n} \sim N(\param_d, P\left[\eif_{\param, d}(P)\right]^2)$ and prior $\theta_d \sim N(0, \tau^2)$. In practice, given an asymptotically normal and efficient estimator $\param_n$ of $\param_0$ with estimated standard errors $\sigma^2_{n}$, we form the Empirical Bayes estimator
\begin{align}
    \penalparam_{d,n}^\mathsf{eb} = \frac{\frac{1}{D - 1} \| \param_n \|_2^2}{\frac{1}{D - 1} \| \param_n \| + \sigma_{d, n}d^2} \param_n.
\end{align}
Confidence intervals can be formed as before, but plugging in the $d$-specific shrinkage factors such that their length adapts to the precision of the estimates of the parameters.

\section{\texorpdfstring{$L_1$}{L1} penalty}
\label{sec:l1-penalty}
In this section we consider penalized parameter defined with an $L_1$ penalty term.
As before, we combine the penalty term with the squared-error loss function $\loss(x, \tilde{x}) = \| x - \tilde{x}\|_2^2$. Let $\penaltyf_1(\tilde{x}) = \lambda \| \tilde{x} \|^1_1$ where $\lambda \geq 0$ is fixed. The objective function is then
\begin{align}
    \objective(x, \tilde{x}) = \| x - \tilde{x} \|_2^2 + \lambda \|\tilde{x}\|_1^1.
\end{align}
That the objective is not differentiable everywhere means we cannot apply Theorem~\ref{thm:penalized-eif} to find an EIF for $\penalparam$, which precludes the type of analysis we were able to conduct in the previous section for the $L_2$ penalty. We proceed instead by noting that the penalized parameter has a closed form solution
\begin{align}
    \penalparam_d(P) = S_\lambda(\param_d(P)),
\end{align}
where $S_\lambda : \mathbb{R} \to \mathbb{R}$ is the soft-thresholding operator
\begin{align}
    x \mapsto S_\lambda(x) = \begin{cases}
        x + \lambda, & x < -\lambda, \\
        0,           & |x| \leq \lambda, \\
        x - \lambda, & x > \lambda.
    \end{cases}
\end{align}
When applied to a vector (i.e. for $S_\lambda : \mathbb{R}^d \to \mathbb{R}^d$) the soft-thresholding operator is to be understood as applying element-wise. This solution shows that the penalized parameter simply shifts the original parameter towards zero by the amount $\lambda$, unless the original parameter is already within $\lambda$ of zero, in which case it is shrunk identically to zero. 

As in the $L_2$ case, we propose a data-driven approach for choosing $\lambda$. Our goal is to pick a $\lambda$ that reduces the finite-sample variance of the penalized parameter with respect to the original parameter. In addition, the $L_1$ penalty may induce a parameter that is sparse, in the sense that it may contain more zeros than the original parameter. We seek an estimator that converges asymptotically to the original parameter by choosing $\lambda$ data-adaptively such that $\lambda$ converges to zero with sample size. 

Our method for choosing $\lambda$ involves approximating the finite-sample bias and variance of an estimator of $\penalparam_\lambda$ depending on the choice of $\lambda$. The non-pathwise differentiability of $\penalparam_\lambda$ in this context precludes the approach we took for the $L_2$-penalized parameter; accordingly, we need to make a bolder approximation. An asymptotically normal and efficient estimator $\param_{d,n}$ of $\param_{d}(P)$, for $d \in \mathcal{D}$, has a limiting distribution given by
\begin{align}
    \sqrt{n} (\param_n - \param(P)) \convd N(0, \sigma^2_{\param, d}(P)).
\end{align}
Based on this, we approximate the finite-sample distribution of $\param_{d,n}$ by the normal distribution:
\begin{align}
    Z_d \sim N\left(\psi_d(P), \frac{1}{n} \sigma^2_{\param, d}(P) \right). 
\end{align}
Suppose that we apply the soft-thresholding operator $S_\lambda$ to $Z_d$, yielding a transformed random variable $S_\lambda(Z_d)$. In  Appendix~\ref{appendix:l1-calculations}, we give closed forms for the mean and variance of $S_\lambda(Z_d)$ as a function of $\lambda$ and the mean and variance of $S_\lambda(Z_d)$, which we denote $\mu_\lambda(\param_d(P), \sigma^2_{\param,d}, n)]$ and $\sigma^2_\lambda(\param_d(P), \sigma^2_{\param,d}, n)$. We propose setting the tuning parameter $\lambda$ to the value $\lambda_n^*$ that minimizes the following criterion:
\begin{align}
    \criterion(\lambda, \param(P), \sigma^2_{\param}(P), n) = \sum_{d=1}^D \left[ \left( \mu_\lambda\left(\param_d(P), \sigma^2_{\param,d}, n\right) - \param_d(P)\right)^2 + \sigma^2_\lambda\left(\param_d(P), \sigma^2_{\param,d}, n\right) \right].
\end{align}
The tuning parameter $\lambda$ is then set to be the minimizer of the above criterion:
\begin{align}
    \label{eq:l1-lambda-optimization}
    \lambda^*(\param(P), \sigma^2_\param(P), n) = \argmin_{\lambda \geq 0} \criterion(\lambda, \param(P), \sigma^2_\param(P), n).
\end{align}
The criterion can be interpreted as an approximation of the mean-squared error of the soft-thresholded estimator relative to the original parameter. 
The minimizer of the above optimization problem does not have a closed form solution; in practice we solve it numerically. 

We propose estimating $\lambda^*$ by the plugin estimator $\lambda_n^* = \lambda_n^*(\param_n, \sigma^2_{\param, n}, n)$ based on estimates $\param_n$ and $\sigma^2_{\param, n}$ of $\param_0$ and $\sigma^2_{\param, 0}$. The estimated $\lambda_n^*$ can then be applied to soft-threshold the initial estimates of $\param_{n}$: 
\begin{align}
    \label{eq:l1-penalized-def}
    \penalparam_{\lambda_n^*} = S_{\lambda_n^*}(\param_{n}).
\end{align}
The following theorem establishes the asymptotic normality and efficiency of the proposed estimator.

\begin{theorem}[Asymptotic normality and efficiency of $\penalparam_{\lambda_n^*}$ for $L_1$-penalization]
    \label{thm:asymptotic-normality-estimated-lambda-l1}
    Let $\param_n$ and $\sigma^2_{\param,n}$ be estimators of $\param_0$ and $\sigma^2_{\param, 0}$, respectively. Let $\lambda_n^*$ and $\penalparam_{\lambda_n^*}$ be defined as in \eqref{eq:l1-lambda-optimization} and \eqref{eq:l1-penalized-def}. Assume each of the following:
    \begin{enumerate}
        \item \label{assumption:l1-one-nonzero} There exists at least one non-zero $\param_{d,0}$: $\|\param_0\|_\infty > 0$. 
        \item \label{assumption:l1-underlying-efficiency} The estimator $\param_{n}$ is an asymptotically normal and efficient:
        \begin{align}
            \sqrt{n}\left( \param_n - \param_0 \right) \convd N\left(0, \sigma_{\param,0}^2 \right).
        \end{align}
        \item \label{assumption:l1-consistent-variance} The estimator $\sigma^2_{\param, n}$ is consistent:
        $\| \sigma^2_{\param, n} - \sigma^2_{\param, 0}\|_\infty = o_P(1)$.
        \item \label{assumption:l1-nearly-minimize} The estimators $\lambda_n^*$ \textit{nearly minimize} the minimization criterion, in the sense that
        \begin{align}
            \criterion(\lambda_n^*, \param_n, \sigma^2_{\param, n}, n) \leq \inf_{\lambda \geq 0}  \criterion(\lambda, \param_n, \sigma^2_{\param, n}, n) + o_P(1).
        \end{align}
    \end{enumerate}
    Then it follows that $\penalparam_{\lambda_n^*}$ is an asymptotically normal and efficient estimator of $\param_0$:
    \begin{align}
        \sqrt{n}\left( \penalparam_{\lambda_n^*} - \param_0 \right) \convd N\left(0, \sigma^2_{\param, 0}\right).
    \end{align}
\end{theorem}
The proof is given in Appendix~\ref{appendix:l1-normality-proof}. Assumption~\ref{assumption:l1-one-nonzero} is necessary only to ensure that the limiting criterion function has a unique minimizer. Otherwise, if all the $\param_{d,0}$ are zero, then the limiting criterion function is constant and any $\lambda \geq 0$ is a minimizer. This assumption could be removed by modifying the criterion to penalize large values of $\lambda$. Assumptions~\ref{assumption:l1-underlying-efficiency} and ~\ref{assumption:l1-consistent-variance} are equivalent to the assumptions for Theorem~\ref{thm:asymptotic-normality-estimated-lambda}. Assumption~\ref{assumption:l1-nearly-minimize} is a weak assumption that we expect to hold in practice. 

Asymptotically valid confidence intervals for the soft-thresholded estimator can be formed using the estimated standard errors for the unpenalized parameter, as in \eqref{eq:basic-confidence-intervals}. 

\section{Simulation Studies}
\label{sec:simulation}

In this section we investigate the finite-sample performance of the proposed $L_1$ and $L_2$ penalized estimators for the first two example parameters: non-parametric linear associations and group-specific average treatment effects. A simulation study for the third example, indirectly standardized outcome ratios, is in Appendix~\ref{sec:simulation-study-3}. Reproduction materials for the simulation studies are available at \url{https://github.com/herbps10/efficient\_penalized\_estimation\_paper}. 

\subsection{Simulation study 1: non-parametric linear association}
\label{sec:simulation-study-1}
In this simulation we directly compare our proposed approach to  penalized regression methods. The target parameter is the scaled non-parametric regression coefficient of Example~1, where for each $d \in \mathcal{D}$, the parameter is $\E_P[\Cov_P(Y, X_d \mid X_{(-d)}) / \E_P[\Var_P(Y\mid X_d)]$. The scaling by the expected variance is introduced such that the parameter is equal to the coefficient $\hat{\beta}_d$ of a main-terms linear regression of $Y$ on $X$, allowing us to directly compare our approach to traditional penalized linear regression estimators.

The simulation setup is a sparse linear regression scenario. Let $X = ( X_1, \dots, X_{100})^T$ be a row vector of covariates, where $X_k \sim \mathrm{Binomial}(0.5)$ for $k = 1, \dots, 100$. Let $\beta \in \mathbb{R}^K$ be a vector of coefficients, and draw
$Y = \beta X + \epsilon$ where $\epsilon \sim N(0, \sigma^2)$. The regression coefficients are fixed at the beginning of each simulation by drawing $\beta_k \sim \mathrm{Binomial}(\theta)$ with $\theta = 30\%$. The simulation study tested all combinations of sample size $N \in \{ 50, 100, 250, 500 \}$ and noise standard deviation $\sigma \in \{ 0.5, 1, 3 \}$. 

To implement the penalized estimators, we need a non-parametric estimator of the non-parametric linear association that is asymptotically normal and efficient. Appendix~\ref{appendix:one-step-example-1-estimator} describes such an estimator based on one-step estimation. The nuisance parameters are estimated using $L_1$-regularized generalized linear regressions with tuning parameters chosen via cross-validation, using the implementation in the \texttt{glmnet} \texttt{R} package \citep{friedman2010regularization, tay2023elasticnet}. The unpenalized estimator is then adjusted using the proposed penalization methods to form $L_1$- and $L_2$-regularized estimators of $\penalparam_d$. 

As a benchmark, we estimated the linear association parameters by fitting $L_1$- and $L_2$-regularized main-terms linear models of $Y$ with respect to covariates $X$ and an intercept term, and take the estimated coefficient $\hat{\beta}_d$ as an estimate of the corresponding linear association parameter $\param_{d,0}$. The tuning parameters were chosen by the default cross-validation method implemented in \texttt{glmnet}. We expect this benchmark estimator to be a consistent estimator of $\param_{d, 0})$ as the simulation data-generating process is a linear model. We compare our approach to the benchmark in terms of the estimates mean error (ME), variance ($\Var$), mean square error (MSE), and 95\% empirical coverage. The comparison method \texttt{glmnet} does not report confidence intervals by default, so we do not compare our method to \texttt{glmnet} in terms of empirical coverage.

A subset of the results corresponding to simulations with noise $\sigma = 3$ are shown in Figure~\ref{fig:simulation-covariance}; a complete table of the results is available as Appendix Table~\ref{tab:simulation-covariance}. Our proposed $L_1$ and $L_2$ penalized estimators match or outperform the unpenalized one-step estimator for all sample sizes and noise levels. The benchmark penalized regressions tended to achieve slightly lower MSE. The better performance of the benchmark in this setting is probably because these methods are tuned using cross-validation, which likely provides better finite-sample approximations of variance than our method, which chooses the strength of penalization parameter $\lambda$ based on an asymptotic approximation.

\begin{figure}[ht!]
    \centering
    \includegraphics[width=1\columnwidth]{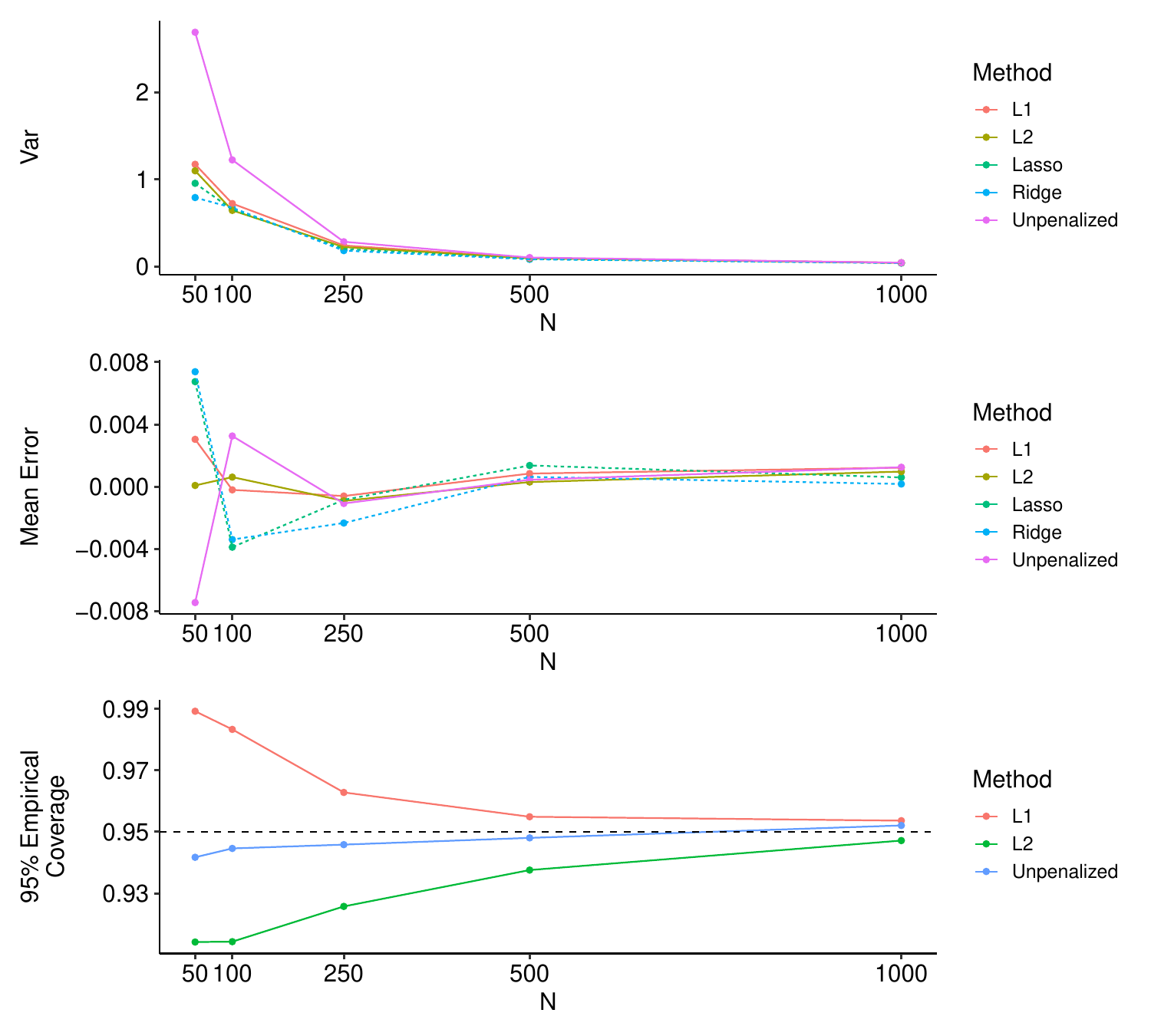}
    \caption{Subset of results from Simulation Study 1 for the non-parametric linear association parameter plotting MSE for all methods with the data-generating process noise size $\sigma = 3$. }
    \label{fig:simulation-covariance}
\end{figure}

\subsection{Simulation study 2: group-specific average treatment effects}
\label{sec:simulation-study-2}
In the second simulation study we investigate estimating group-specific average treatment effects. The simulation data-generating process is as follows: first, a treatment effect is drawn for each of the population subgroups. For $D > 0$ subgroups, treatment effects are set as $\beta_d = \delta_d \times \alpha_d$, where $\delta_d \sim \mathrm{Binomial}(\theta)$ and $\alpha_d \sim \mathrm{Uniform}(-1, 1)$. The parameter $\theta \in [0, 1]$ controls the probability of a group having a non-zero treatment effect. Next, $N > 0$ observations $O = (X, G, A, Y)$ are independently drawn where
\begin{itemize}
    \item $X = (X_1, \dots, X_5)$ is a vector of covariates with $X_k \sim \mathrm{Unif}(0, 1)$ for $k = 1, \dots, 5$,
    \item $G \in \{1, \dots, D\}$ is a group-membership indicator drawn uniformly at random,
    \item $A$ is a binary treatment variable drawn according to the law
    \begin{align}
        A \sim \mathrm{Binomial}\left(\mathrm{logit}^{-1}\left( X_1 + \alpha_d - \alpha_d X_2\right)\right)
    \end{align}
    \item $Y \in \mathbb{R}$ is a continuous outcome drawn according to the law
    \begin{align}
        Y \sim \mathrm{Normal}\left(2 X_1 - 2X_2 + 0.5 X_5^2 + \beta_G A, \sigma^2\right). 
    \end{align}
\end{itemize}
The parameters that we vary in the simulation study are $N \in \{4000, 6000, 8000, 10000 \}$, the total number of observations across all groups; $\theta \in \{ 0\%, 30\%, 100\% \}$, the probability of a group having a non-zero treatment effect; and $\sigma \in \{ 1, 2, 4 \}$, the outcome noise standard deviation. The number of groups is set to $G = 25$ for all simulations. Every combination of the aforementioned parameters are tested by independently simulating $250$ datasets from the simulation data-generating process. 

For each simulated data set we first estimate the non-penalized group-specific ATE by applying an estimator based on Targeted Maximum Likelihood Estimation (TMLE) separately to the observations from each group. We use the TMLE algorithm implemented in the \texttt{tmle} 
 \texttt{R} package \citep{gruber2012tmle}. The nuisance parameters (propensity score and outcome model) are estimated using an ensemble method (Super Learner) that incorporates generalized linear models with main terms, generalized linear models with interactions, and regularized linear models (\texttt{SL.glm}, \texttt{SL.SL.glm.interaction}, and \texttt{SL.glmnet} learners, respectively). The TMLE algorithm provides both point estimates and standard errors for each of the group-specific ATEs. We then 
 apply our proposed $L_1$ and $L_2$ regularization adjustments to form estimates of the penalized parameters $\penalparam_d$. 

We compare the $L_1$ and $L_2$ regularized estimates to the original unpenalized estimates in terms of the mean error (ME), mean squared error (MSE) and the empirical coverage of the 95\% confidence intervals, averaged across the $D$ group-specific ATE estimates. The results are shown in Table~\ref{tab:simulation-group-effects}. Particularly at the smallest sample size ($N = 2000$) and largest outcome noise standard deviation ($\sigma = 4$), the $L_1$ and $L_2$ penalized parameters had lower MSE than the unpenalized estimates. Interestingly, the $L_1$ and $L_2$ penalized estimates tended to have smaller mean error than the unpenalized estimator, suggesting that penalization did not incur a bias-variance trade-off penalty. The confidence intervals for the unpenalized point estimates achieved near-optimal 95\% empirical coverage in all scenarios. The confidence intervals based on the penalized and shrinkage point estimates tended to be conservative, especially with when the noise was high. 

\begin{table}[ht]
    \centering
    \small
    \begin{tabular}{|lrrrrrrrrrrrrr|}
    \hline
    \multicolumn{2}{|l}{} & \multicolumn{4}{c}{MSE} & \multicolumn{4}{c}{ME} & \multicolumn{4}{c|}{95\% Empirical Coverage} \\
    $\sigma$ & $N$ & $\psi_n$ & $L_1$ & $L_2$ & EB & $\psi_n$ & $L_1$ & $L_2$ & EB & $\psi_n$ & $L_1$ & $L_2$ & EB \\
    \hline
    0.5 & 4000 & 0.8 & 0.5 & 0.7 & 0.7 & -0.3 & -0.3 & -0.3 & -0.3 & 94.2\% & 96.0\% & 92.8\% & 92.9\%\\
      & 6000 & 0.5 & 0.3 & 0.5 & 0.5 & 0.1 & 0.1 & 0.1 & 0.2 & 94.0\% & 95.9\% & 92.9\% & 93.0\%\\
      & 8000 & 0.4 & 0.2 & 0.3 & 0.3 & -0.1 & -0.1 & -0.1 & 0.0 & 95.1\% & 96.9\% & 94.7\% & 94.6\%\\
      & 10000 & 0.3 & 0.2 & 0.3 & 0.3 & 0.0 & 0.0 & 0.0 & 0.0 & 94.6\% & 96.6\% & 94.4\% & 94.4\%\\
     1 & 4000 & 3.1 & 1.9 & 2.4 & 2.4 & -0.2 & -0.1 & -0.1 & 0.0 & 93.7\% & 96.8\% & 91.5\% & 91.8\%\\
      & 6000 & 1.9 & 1.2 & 1.6 & 1.6 & 0.0 & 0.0 & 0.0 & 0.1 & 94.9\% & 97.1\% & 92.8\% & 93.0\%\\
      & 8000 & 1.4 & 0.9 & 1.2 & 1.2 & -0.1 & -0.1 & -0.1 & 0.0 & 95.0\% & 97.0\% & 93.3\% & 93.5\%\\
      & 10000 & 1.2 & 0.8 & 1.0 & 1.1 & -0.2 & -0.1 & -0.2 & -0.1 & 94.2\% & 96.6\% & 93.2\% & 93.2\%\\
     2 & 4000 & 11.8 & 6.1 & 6.3 & 6.4 & -0.2 & -0.1 & -0.3 & -0.1 & 94.3\% & 97.9\% & 91.0\% & 91.1\%\\
      & 6000 & 8.1 & 4.3 & 4.8 & 4.8 & -0.7 & -0.6 & -0.6 & -0.4 & 94.0\% & 97.0\% & 91.1\% & 91.1\%\\
      & 8000 & 5.9 & 3.3 & 3.8 & 3.9 & -0.3 & 0.0 & -0.2 & 0.0 & 94.2\% & 97.3\% & 91.0\% & 91.2\%\\
      & 10000 & 4.7 & 2.8 & 3.3 & 3.3 & -0.3 & -0.1 & -0.2 & 0.0 & 94.6\% & 97.1\% & 91.3\% & 91.4\%\\
     4 & 4000 & 47.8 & 17.7 & 17.5 & 18.0 & -2.5 & -1.3 & -1.3 & -1.0 & 94.1\% & 98.9\% & 91.6\% & 91.9\%\\
       & 6000 & 31.5 & 12.7 & 12.5 & 12.9 & -0.5 & -1.1 & -0.9 & -0.6 & 94.7\% & 98.7\% & 92.4\% & 92.5\%\\
      & 8000 & 23.2 & 9.9 & 9.9 & 10.1 & -0.9 & -0.9 & -0.7 & -0.5 & 94.9\% & 98.7\% & 91.9\% & 91.9\%\\
      & 10000 & 18.7 & 8.7 & 8.7 & 8.9 & -0.4 & -0.1 & -0.2 & 0.1 & 94.5\% & 98.4\% & 90.8\% & 91.0\%\\
    \hline
    \end{tabular}
    \caption{Subset of results from Simulation Study 2 for group-specific ATEs showing mean squared error (MSE), mean error (ME), and empirical 95\% coverage for the unpenalized TMLE estimator $\psi_n$, $L_1$-regularized estimator, and $L_2$-regularized estimator, and Empirical Bayes (EB) shrinkage estimator where the probability of positive group-specific treatment effect $\theta = 30\%$ and varying outcome noise standard deviations $\sigma$, and overall sample sizes $N$. Additional results are available as Appendix Table~\ref{tab:simulation-group-effects-supplemental}.}
    \label{tab:simulation-group-effects}
\end{table}

\section{Application}
\label{sec:application}

In this section we illustrate the real-world utility of our penalization methods through a healthcare provider profiling application, estimating the standardized readmission ratios (SRR) for kidney dialysis providers. Briefly, the observed data are a set of baseline patient covariates $X$, a treatment variable $A \in \{1, \dots, D \} = \mathcal{D}$ that indexes the dialysis provider seen by each patient, and an outcome variable $Y \in \{0, 1\}$ which indicates all-cause unplanned hospital readmission within 30 days of discharge ($Y = 1$ indicates unplanned readmission, which is considered a negative outcome). Define the indirectly-standardized outcome $\param_d$ for a provider $d \in \mathcal{D}$ as in Example~3. That is, $\param_d$ is (under causal assumptions) the mean unplanned readmission rate if the population of patients treated by provider $d$ had rather been randomly assigned to another provider according to the observed provider-assignment mechanism. We then define the centered standardized readmission ratio (SRR) as the ratio of $\param_d$ to the observed readmission rates for patients treated by provider $d$, centered at zero:
\begin{align}
    \label{eq:srr}
    \mathsf{SRR}_d(P) := \frac{\param_d(P)}{\E_P[Y \mid A = d]} - 1.
\end{align}
A positive SRR means that the unplanned readmission rate would have been higher if patients had been randomly assigned to a provider that treated a similar patient mix; this can be seen as evidence of better performance of provider $d$ relative to its peers treating a similar population. Similarly, a negative SRR suggests that the unplanned readmission rate would have been lower if patients were randomly reassigned to another provider. 

Estimating the above SRR parameter may be difficult, especially for providers with few patients. In addition, there are typically high policy stakes involved in provider profiling, as the results may be used to identify under-performing providers for remedial action. Thus, there is often interest in having any estimates be conservative by shrinking high-variance estimates towards zero. This approach avoids unfairly penalizing small providers who, for example, purely by chance happened to have treated patients who had a unusually high number of unplanned readmissions. 

A popular approach for estimating provider profiling measures with shrinkage is via generalized mixed models with a provider-specific random effect that is shrunk towards zero. However, as explored in simulations in \citet{susmann2024doublyrobustnonparametricefficient}, generalized linear models introduce parametric assumptions on the data-generating process that can lead to biased estimates. In addition, we argue that shrinking the actual parameter of interest, the SRR, towards zero is more interpretable than shrinking the provider-specific random effects of a generalized linear model, which have a complex interpretation. 

We analyze data from a Medicare claims dataset from the United States Renal Data System (USRDS) consisting in dialysis provider treatment records for patients with end-stage renal disease (ESRD) \citep{usrds2022}. These data were previously analyzed in \citet{susmann2024doublyrobustnonparametricefficient}, in which non-penalized SRRs were estimated using doubly robust and asymptotically consistent estimators. Our analysis dataset comprises all dialysis providers in New York State with at least 20 observations in the year 2020 (this enlarges our previous analysis of the same data, which used only those providers with at least 100 observations). We compare estimates of the non-penalized SRR, as in the previous study, to estimates of $L_2$-penalized SRR with penalization parameter $\lambda$ chosen using the data-driven criterion proposed in Section~\ref{sec:ridge}. We also applied the Empirical Bayes shrinkage derived in Section~\ref{sec:ridge} that adaptively shrinks estimates as a function of the standard error. 

Results from the applied analysis are shown in Figure~\ref{fig:srr-funnel-plots}. The results are displayed as funnel plots, which plot the precision of the unpenalized SRR estimator vs. the SRR point estimates, before and after adjustment. A notable difference in the estimates adjusted by $L_2$ penalization versus Empirical Bayes shrinkage is in the high-precision estimates. As expected, the $L_2$ penalization is based on a single penalization parameter $\lambda$, which causes all parameters to be shrunk towards one, including the high-precision estimates. This is not true of the Empirical Bayes estimates, which are shrunk less for high-precision estimates. 

\begin{figure}[ht!]
    \centering
    \includegraphics[width=0.75\columnwidth]{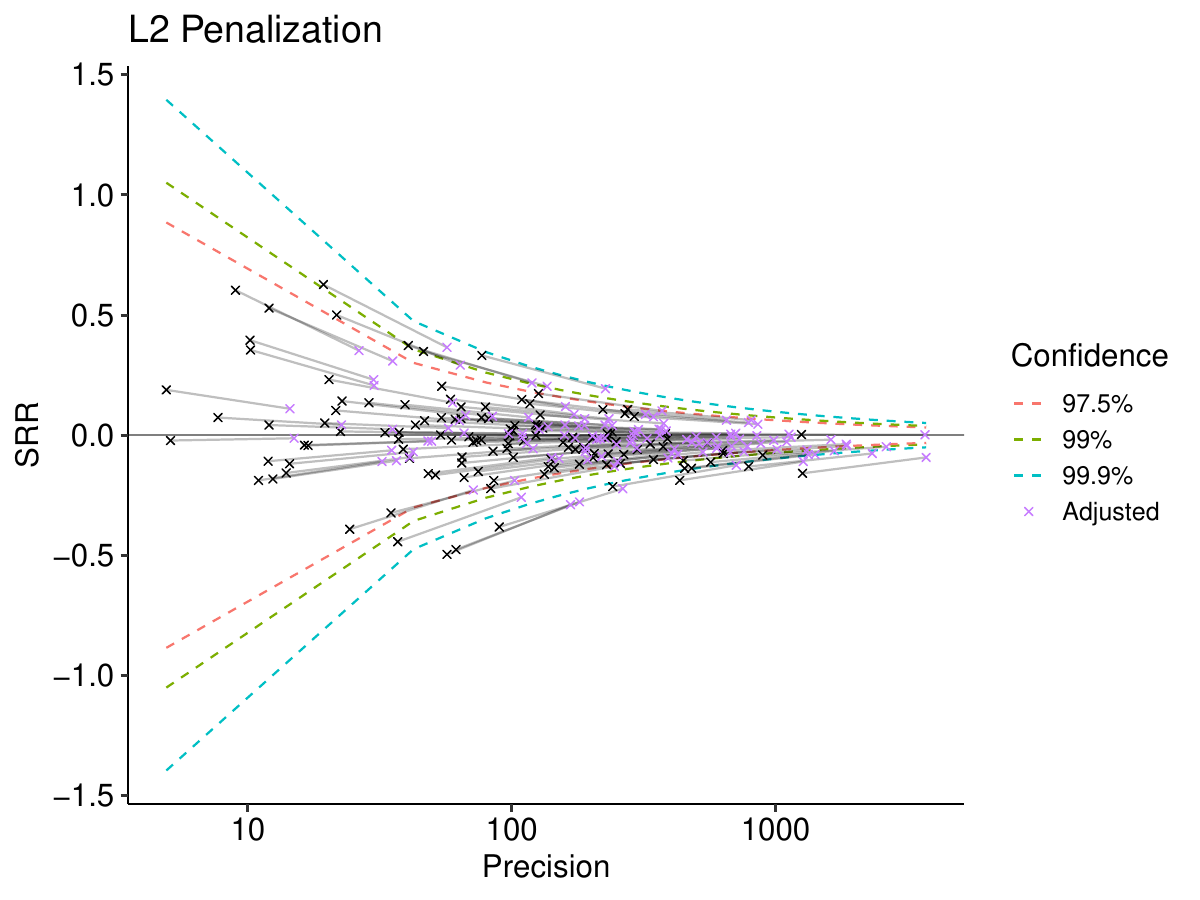}\\
    \includegraphics[width=0.75\columnwidth]{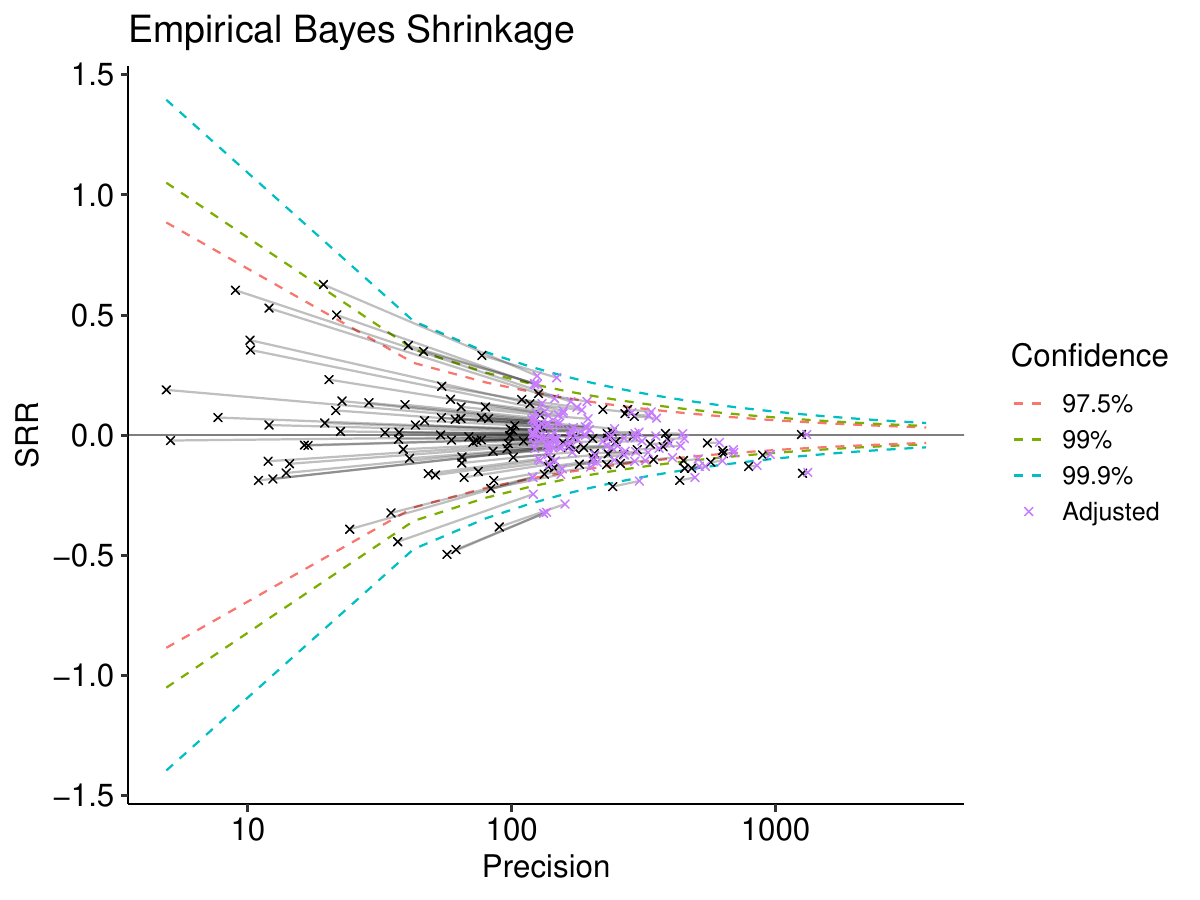}
    \caption{Funnel plots of Standardized Readmission Ratios (SRR, Section \ref{eq:srr}) all for New York State dialysis providers in the analysis dataset. In the top plot SRR estimates are adjusted by $L_2$ penalization using the data-adaptive choice of penalization hyperparameter $\lambda$ proposed in Section~\ref{sec:ridge} and shrinkage standard errors. In the bottom plot SRR estimates are shrunk using the Empirical Bayes based method described in Section~\ref{sec:ridge}. Vertical lines connect dialysis provider SRR estimates before and after adjustment.}
    \label{fig:srr-funnel-plots}
\end{figure}

\section{Discussion}
\label{sec:discussion}

Estimating a large set of statistical parameters introduces challenges beyond those of estimating a single parameter. To improve estimation, it may be of interest to trade bias in one of the constituent parameters in favor of controlling the overall variance across all estimates. In addition, to aid interpretation or communication it may also be of interest to find a set of point estimates that are \textit{sparse}, in that estimates statistically indistinguishable from zero are shrunk identically to zero. To address these concerns, we introduced a novel framework for defining regularized statistical parameters via penalization. This framework maintains the substantive focus on the original parameter of interest, and penalized parameters are introduced as a way to derive estimators with desirable finite properties such as lower variance and sparsity.

The penalized parameters we propose are formulated in a completely non-parametric framework, and our results are therefore applicable in very general settings. One particular area where they are relevant is causal inference, where the target parameters of interest are typically a (possibly large) set of low-dimensional summaries of counterfactual quantities, such as treatment effects. While existing methods such as penalized regression can be applied to estimate the nuisance parameters required for forming efficient and doubly-robust causal effect estimators, it is less clear how to apply penalization directly to the causal effect estimates themselves. Our research fills this gap.

We explored two important examples of penalized parameters that fall within our framework: those defined with $L_2$ and $L_1$ penalization terms. Many other options are available; considering $L_p$-norm penalties in more generality would be an immediate extension, or penalties such as the Elastic-Net penalty or the Huber loss function. Going further, our framework could be expanded to capture functional parameters (such as those in a Banach or Hilbert Space) regularized with functional norms.

Within the $L_2$ and $L_1$ examples we investigated, our proposed data-adaptive approaches for choosing the penalization hyperparameter $\lambda$ are based on an asymptotic approximation of the variance of the unpenalized estimators. For the $L_2$ penalty example, for example, we use asymptotically-justified variance approximations with the goal of forming an estimator with better \textit{finite-sample} performance. The reliance on asymptotic approximations is due to the generality of our approach in which the key restriction is the pathwise differentiability of the original parameter, the property that leads to the existence of an EIF for the target parameter characterizing its efficiency bound. We then use this asymptotic efficiency bound to approximate finite-sample variance. However, other methods for choosing $\lambda$ may perform better than our approach, in particular when finite-sample variance expressions are available or cross-validation can be applied. Indeed, results from the first simulation study show that penalized linear regression tuned with cross-validation can yield lower MSEs than our proposed estimators. However, the applicability of cross-validation in that context hinges on the fact that the parameter of interest is identified as a linear regression coefficient. Cross-validation can then be applied to find an optimal degree of penalization based on the model's predictive performance. However, for the other causal target parameters we investigated, it is not clear how cross-validation could be so straightforwardly applied as the target parameters are low-dimensional summaries of counterfactuals, and are not predictive. The strength of our approach, then, is its general applicability to low-dimensional target parameters, such as those of interest in causal inference that are typically defined in terms of counterfactual quantities. 

Statistical inference based on the penalized estimators a challenging problem. In this work, we assumed a scenario in which substantive interest lies in the original, non-penalized parameter. The goal of introducing penalization is then a tool for improving the finite-sample properties of an estimator; the penalized parameter \textit{itself} is not of substantive interest. For example, scientific interest typically lies in estimating a set of group-specific average treatment effects, and finding valid confidence intervals for those treatment effects; the goal is not to forming valid confidence intervals for a \textit{penalized} treatment effect. For this reason, we presented results showing that, based on our data-adaptive proposals for the choice of the penalization parameter $\lambda$, the penalized estimates converge asymptotically to the original, non-penalized parameter. In the $L_2$ case we found that confidence intervals for the penalized estimator perform well as confidence intervals for the original parameter in finite-sample simulations. For the $L_1$ case, asymptotically valid intervals can be formed based on the estimated variance of the original parameter. Such intervals are asymptotically valid, but not entirely satisfying given that they do not shrink as the $L_2$ intervals do. Further research could address alternative methods to build confidence intervals; adapting recent developments from the Empirical Bayes literature is one promising avenue \citep{armstrong2022empiricalbayes, jiaying2023comparisons}.

\subsection*{Acknowledgments} We would like to thank Antoine Chambaz and Alec McClean for helpful discussions. The computational requirements for this work were supported in part by the NYU Langone High Performance Computing (HPC) Core's resources and personnel. The data reported here have been supplied by the United States Renal Data System (USRDS). The interpretation and reporting of these data are the responsibility of the author(s) and in no way should be seen as an official policy or interpretation of the
U.S. government.

\bibliography{references}

\appendix
\counterwithin{figure}{section}
\counterwithin{table}{section}

\section{Appendix}

\subsection{Notation reference}
\label{appendix:notation}

\begin{table}[ht]
    \centering
    \begin{tabular}{|l|p{15cm}|}
        \hline
        Symbol & Definition \\
        \hline
         $\model$ & Non-parametric statistical model \\
         $P$ & A distribution in $model$. \\
         $P_0$ & The data-generating distribution.  \\
         $\mathcal{D}$ & Set indexing parameters of interest
         $\param$ A vector-valued parameter $\param : \model \to \mathbb{R}^{|\mathcal{D}|}$. \\
         $\penalparam_\lambda$ & Penalized parameter defined in terms of $\param$ with tuning parameter $\lambda$. \\
         $\lambda$ & Penalization tuning parameter. \\
         $U_\lambda$ & Minimization objective function in the definition of $\penalparam_\lambda$ \eqref{eq:penalized-parameter-def}. \\
         $L$ & Loss function term in the objective function $U_{\lambda}$. \\
         $V_\lambda$ & Penalty term in the objective function $U_{\lambda}$.  \\
         $D_\phi(P)$ & Efficient Influence Function (EIF) of a parameter $\phi$ at $P$. \\
         $\sigma^2_\phi(P)$ & Variance of the EIF of the parameter $\phi$ evaluated at $P$; defines the efficiency bound for estimating $\phi$ in a non-parametric model. \\
         $\dot{U}_\lambda$ & Derivative of $U_{\lambda}(x, \tilde{x})$ with respect to its second argument; defined in Theorem~\ref{thm:penalized-eif}. \\
         $\ddot{U}_\lambda$ & Derivative of $\dot{U}_\lambda(x, \tilde{x}$ with respect to its second argument; defined in Theorem~\ref{thm:penalized-eif}. \\
         $\nabla \dot{U}_\lambda$ & Derivative of $\dot{U}_\lambda(x, \tilde{x}$ with respect to its first argument; defined in Theorem~\ref{thm:penalized-eif}. \\
         $R$ & Second-order remainder term of von-Mises expansion \ref{eq:von-mises}. \\
         $\criterion$ & Minimization objective for choosing data-adaptive tuning parameter. \\
         \hline
    \end{tabular}
    \caption{Key notation used in the manuscript and appendix.}
    \label{tab:my_label}
\end{table}

\subsection{One-step estimation}
\label{appendix:one-step}
One strategy for constructing an estimator, referred to as \textit{one-step estimation}, relies on analysis of a von-Mises expansion \eqref{eq:von-mises}. Suppose we have an initial estimate $P_n^0$ of the parts of $P$ relevant to the parameter $\phi$ and $\eif_\phi$. Setting $P_1 = P_n^0$ and $P_2 = P_0$ in \eqref{eq:von-mises}, we have
\begin{align}
    \phi(P_n^0) - \phi_0 = -P_0 D_\phi(P_n^0)  + R_2(P_n^0, P_0).
\end{align}
The initial plug-in estimator $\phi(P_n^0)$ therefore has first-order bias equal to the mean of the EIF evaluated at the initial estimates $P_n^0$, and second-order bias given by $R(P_n^0, P_0)$. This suggests forming a one-step estimator by adding the empirical mean of the EIF to the initial estimates:
\begin{align}
    \label{eq:phi-one-step}
    \phi^{\mathsf{os}} = \phi(P_n^0) + P_n \eif_\phi(P_n^0). 
\end{align}
This estimator is referred to as a one-step estimator as it can be thought of as a type of one-step Newton correction to the original estimator. In addition, it can be thought of as a serving as a type of non-parametric analog to Le Cam's one-step method for parametric models. 

A classical approach to establish that the one-step estimator is asymptotically normal and efficient requires placing complexity conditions (such as Donsker conditions) on the nuisance estimators used to form the initial estimate $P_n^0$. However, such assumptions can be obviated through the use of cross-fitting \citep{bickel1982,schick1986,vanderVaart98,zheng2010crossfitting}. First, split the observed data $O_1, \dots, O_n$ into $1 < K < \infty$ disjoint folds by drawing $n$ i.i.d. draws $Z_1, \dots, Z_n$ of a categorical random variable $Z \in \{1, \dots, K\}$, where $Z_i = k$ indicates that observation $i$ belongs to fold $k$. Let $P_k^0$ be an initial estimate of the parts of $P$ relevant to $\phi$ and $\eif_\phi$ estimated only using observations not in the fold $k$, and denote by $P_n^k$ the empirical measure over the observations within the fold $k$. Then the analog of the one-step estimator \eqref{eq:phi-one-step} for fold $k$ is given by
\begin{align}
    \phi^{\os}_k = \phi(P_k^0) + P_n^k \eif_\phi(P_k^0).
\end{align}
The final estimator is the average of the fold-specific one-step estimators:
\begin{align}
    \phi^{\os} = \sum_{k=1}^K \frac{N_k}{n} \phi^{\os}_k,
\end{align}
where $N_k = \sum_{i=1}^n Z_i$ be the number of observations in fold $k$.  
Consistency, asymptotic normality, and efficiency of the cross-fitted one-step estimator can be established under suitable conditions on the estimates $P_k^0$ and on the rate of convergence to zero of the cross-fitted remainder term. We state general versions of the assumptions below.
\begin{assumption}[Consistent estimation of EIF]
    \label{assumption:consistent-eif-estimation}
    Assume that $\| \eif_\phi(P_k^0)(O) - \eif_\phi(P)(O) \| = o_P(1)$ for each fold $k \in \{1, \dots, K\}$. 
\end{assumption}
\begin{assumption}[Rate of convergence of remainder]
    \label{assumption:rate-of-convergence}
    Assume that $\sum_{k=1}^K \frac{N_k}{n} R_2(P_k^0, P) = o_P(1/\sqrt{n})$ for each fold $k \in \{1, \dots, K \}$. 
\end{assumption}
In practice, the specific form of the EIF and second-order remainder term corresponding to a particular penalized parameter will typically imply more granular assumptions on the nuisance estimators used to form $P_n^0$. 

With the form of the EIF for $\penalparam_\lambda$ in hand, a one-step estimator of $\penalparam_\lambda(P)$ can be formed following the description in the previous section. Specifically, we need fold-specific initial estimates $P_k^0$ of the parts of $P$ relevant to $\penalparam_\lambda(P)$ and $\eif_{\penalparam_\lambda}(P)$. Within each fold, $\penalparam_{\lambda}(P_k^0)$ can be found by solving the optimization problem \eqref{eq:penalized-parameter-def}. The cross-fitted one-step estimator is then as defined in \eqref{eq:phi-one-step}. The following theorem states conditions under which the resulting estimator is asymptotically normal and efficient. 
\begin{theorem}[Asymptotic normality and efficiency of one-step estimator $\penalparam_{\lambda, n}^{\os}$]
    \label{thm:general-asymptotic-normality}
    Assume Assumptions~\ref{assumption:consistent-eif-estimation}, and \ref{assumption:rate-of-convergence} for the fold-specific initial estimates $P_k^0$. 
    Then the cross-fitted estimator $\penalparam_{\lambda, n}^{\os}$ is asymptotically normal and efficient:
    \begin{align}
        \sqrt{n}\left( \penalparam_{\lambda, n}^{\mathsf{os}} - \penalparam_{\lambda, 0} \right) \convd N\left(0, \sigma^2_{\penalparam_\lambda, 0}\right).
    \end{align}
\end{theorem}
The proof follows straightforwardly from \citealt[Proposition 2]{kennedy2024review}. Theorem~\ref{thm:general-asymptotic-normality} provides high-level results for one-step estimators of any pathwise differentiable penalized parameter. In the following sections, we specialize to specific loss functions and penalty terms, which also allows us to establish more granular conditions under which one-step estimation is asymptotically normal and efficient.

\subsection{One-step estimation of non-parametric linear regression parameter}
\label{appendix:one-step-example-1-estimator}
We estimate the nuisance parameters $\propscore_P$ and $\outcomemodel_P$ via cross-fitting with $K$ folds. Within each cross-fitting fold $k$, an estimate of the parameter $\param_d$ is formed as
\begin{align}
    \param^{\os}_{d,k} = P_n^k \left[ \left(X_d - \hat{\propscore}_{k}(X_{(-d)}) \right) \left( Y - \hat{\outcomemodel}_k(X_{(-d)} \right) \right].
\end{align}
The final estimate is formed by averaging the estimates from the $K$-folds:
\begin{align}
    \hat{\param}_d^{\os} = \sum_{k=1}^K \frac{N_k}{n} \param_{d,k}^{\os}. 
\end{align}

\section{Additional Proofs}
\label{appendix:l2-penalized-approximation}

\subsection{Proof of Theorem~\ref{thm:alternative-decomposition}}
\label{appendix:l2-expansion-proof}
\begin{proof}
The assumption that $\param$ satisfies the von-Mises expansion \eqref{eq:von-mises} implies we can write, for any $P_1, P_2 \in \model$,
\begin{align}
    \label{eq:repeat-param-von-mises}
    \param(P_1) - \param(P_2)  = - P_2\left[ \eif_{\param}(P_1) \right] + R(P_1, P_2). 
\end{align}
A similar expansion for $\penalparam_{\lambda^*}$ would take the form:
\begin{align}
    \label{eq:repeat-l2-penalized-von-mises}
    \penalparam_{\lambda_n^*}(P_1) - \penalparam_{\lambda_n^*}(P_2) = -P_2\left[ \eif_{\penalparam_{\lambda_n^*}}(P) \right] + R_{\penalparam}(P_1, P_2).
\end{align}
Recall from Theorem~\ref{thm:asymptotic-normality-estimated-lambda} that the EIF of $\penalparam_{\lambda^*}$ is given by
\begin{align}
    \eif_{\penalparam_{\lambda^*}}(P)(O) = \frac{1}{1 + \lambda^*(P)} \eif_{\param}(P)(O) - \frac{1}{n} \times \frac{\param(P)}{(1 + \lambda^*(P))^2} \eif_{\gamma}(P)(O).
\end{align}
Decompose the above EIF into two parts, such that $\eif_{\penalparam_{\lambda^*}}(P)(O) = \eif_{1}(P)(O) + \frac{1}{n} \eif_{2}(P)(O)$, with
\begin{align}
    \eif_1(P)(O) &= \frac{1}{1 + \lambda^*(P)} \eif_{\param}(P)(O), \\
    \eif_2(P)(O) &= -\frac{1}{n} \times \frac{\param(P)}{(1 + \lambda^*(P))^2} \eif_{\gamma}(P)(O).
\end{align}
The expansion \eqref{eq:repeat-l2-penalized-von-mises} can then be rewritten as
\begin{align}
    \label{eq:l2-expansion-rewrite-1}
    \penalparam_{\lambda^*}(P_1) - \penalparam_{\lambda^*}(P_2) = -P_2\left[ \eif_1(P_1) \right] - \frac{1}{n} P_2\left[ \eif_2(P_1) \right] + R_{\penalparam}(P_1, P_2). 
\end{align}
Next, analyze the form of the remainder term $R_{\penalparam}(P_1, P_2)$:
\begin{align}
    R_{\penalparam}(P_1, P_2) &= \penalparam_{\lambda^*}(P_1) - \penalparam_{\lambda^*}(P_2) + P_2\left[ \eif_{\penalparam_{\lambda^*}}(P_1) \right]  \qquad \text{(by \eqref{eq:repeat-l2-penalized-von-mises})} \\
    &= \frac{1}{1 + \lambda^*(P_1)} \param(P_1) - \frac{1}{1 + \lambda^*(P_2)} \param(P_2) + P_2\left[ \eif_{\penalparam_{\lambda^*}}(P_1) \right] \qquad \text{by def'n of $\penalparam_{\lambda^*}$)} \\
    &= \frac{1}{1 + \lambda^*(P_1)} \left\{  \param(P_2) - P_2\left[ \eif_{\param}(P_1) \right] + R(P_1, P_2) \right\} \\
    &\quad - \frac{1}{1 + \lambda^*(P_2)} \param(P_2) + P_2\left[ \eif_{\penalparam_{\lambda^*}}(P_1) \right] \qquad \text{(by \eqref{eq:repeat-param-von-mises})} \\
    &= \left\{ \frac{1}{1 + \lambda^*(P_1)} - \frac{1}{1 + \lambda^*(P_2)} \right\} \param(P_2) + \frac{1}{n} P_2\left[ \eif_2(P_1) \right] + \frac{1}{1+\lambda^*(P_1)} R(P_1, P_2). 
\end{align}
Combining the above with the expansion \eqref{eq:l2-expansion-rewrite-1} yields the result:
\begin{align}
    \penalparam_{\lambda^*}(P_1) - \penalparam_{\lambda^*}(P_2) = -P_2\left[ \eif_1(P_1) \right] + \left\{ \frac{1}{1 + \lambda^*(P_1)} - \frac{1}{1 + \lambda^*(P_2)} \right\} \param(P_2) + \frac{1}{1+\lambda^*(P_1)} R(P_1, P_2). 
\end{align}
\end{proof}

\subsection{Proof of Theorem~\ref{thm:asymptotic-normality-estimated-lambda-l1}}
\label{appendix:l1-normality-proof}
\begin{proof}
    The continuity $\mu_\lambda$ and $\sigma^2_\lambda$ can be readily seen based on their definition given in Appendix~\ref{appendix:l1-calculations}. 
    By the continuous mapping theorem, the assumed consistency of the estimators $\param_n$ and $\sigma^2_{\param, n}$ implies that for all $d \in \mathcal{D}$ and $\lambda \geq 0$,
    \begin{align}
        \mu_\lambda(\param_{d,n}, \sigma^2_{\param, d, n}, n) &\convp \mu_\lambda(\param_{d}, \sigma^2_{\param, d, 0}), \\
        \sigma^2_\lambda(\param_{d,n}, \sigma^2_{\param, d, n}, n) &\convp 0,
    \end{align}
    because, by the definition of $\sigma^2_\lambda(\param_{d}, \sigma^2_{\param, d}, n)$ as $n \to \infty$ then
    \begin{align}
        \sigma^2_\lambda(\param_{d}, \sigma^2_{\param, d}, n) \to 0.
    \end{align}
    and where 
    \begin{align}
        \mu_{\lambda}(\param_{d, 0}, \sigma^2_{\param, d, 0}) &= \begin{cases}
           0, & |\param_{d,0}| \leq \lambda,\\
           \param_{d,0} + \lambda, & \param_{d,0} < -\lambda, \\
           \param_{d,0} - \lambda, & \param_{d,0} > \lambda.
        \end{cases}
    \end{align}
    Therefore, the random criterion function converges uniformly in $\lambda$ to a limit criterion function:
    \begin{align}
        \sup_{\lambda > 0} \left\| \criterion(\lambda, \param_n, \sigma^2_{\param, n}, n) - \criterion_\infty(\lambda, \param_0) \right\| \convp 0,
    \end{align}
    where
    \begin{align}
        \criterion_\infty(\lambda, \param_0) = \sum_{d=1}^D \left( \mu_\lambda(\param_{d, 0}, \sigma^2_{\param, d, 0}) - \param_{d, 0} \right)^2. 
    \end{align}
    By assumption at least one of the $\param_{d,0}$ is non-zero; therefore, the limiting criterion function has a unique minimum at $\lambda = 0$, as then $\mu_\lambda(\param_{d, 0}, \sigma^2_{\param, d, 0}) - \param_{d,0} = 0$ for all $d \in \mathcal{D}$. Furthermore, the minimizer is well-separated in the sense that for any $\epsilon > 0$, 
    \begin{align}
        \sup_{\lambda \geq \epsilon} \criterion_\infty(\lambda, \psi_0) > 0,
    \end{align}
    which is because for any $\lambda > 0$, there is a $d \in \mathcal{D}$ such that $(\mu_{\lambda}(\param_{d, 0}, \sigma^2_{\param,d,0}) - \param_{d,0})^2 > 0$. 
    Therefore, by \citealt[Theorem 5.7]{vanderVaart98}, $\lambda_n^* \convp 0$. 
    The final result then follows straightforwardly:
    \begin{align}
        \sqrt{n}\left( \penalparam_{\lambda_n^*} - \param_0 \right) &= \sqrt{n}\left( S_{\lambda_n^*}(\param_n) - \param_0 \right) \\
        &\convd \sqrt{n}\left( S_{0}(\param_n) - \param_0 \right) \text{ (continuous mapping theorem) } \\
        &= \sqrt{n}\left(\param_n - \param_0 \right) \\
        &\convd N\left(0, \sigma_{\param, 0}^2 \right) \text{ (by Assumption~\ref{assumption:l1-underlying-efficiency})}. 
    \end{align}
\end{proof}

\section{Additional derivations for \texorpdfstring{$L_1$}{L1}-penalized tuning parameter}
\label{appendix:l1-calculations}

Suppose that a random variable $Z \sim N(\mu, \sigma^2)$. Consider the random variable $S_\lambda(Z)$, which we say follows a soft-thresholded normal distribution with parameters $\lambda$, $\mu$, and $\sigma^2$, written $S_\lambda(Z) \sim N_{\lambda}(\mu, \sigma^2)$. The mean and variance of $S_\lambda(Z)$ have non-trivial relationships with $\mu$ and $\sigma$. Let  $x \mapsto \Phi_{\mu,\sigma^2}(x)$ and $x \mapsto \Phi'_{\mu,\sigma^2}(x)$ be the CDF and PDF of the normal distribution with parameters $\mu$ and $\sigma^2$, respectively.
\begin{theorem}
The mean and variance of $S_\lambda(Z) \sim N_\lambda(\mu, \sigma^2)$ are given by
\begin{align}
    \label{eq:soft-threshold-distribution}
    \E[S_\lambda(Z)] =& \mu - \mu \left( \Phi_{\mu, \sigma^2}(\lambda) - \Phi_{\mu, \sigma^2}(-\lambda) \right) \\
    &+ \lambda \left( \Phi_{\mu, \sigma^2}(\lambda) + \Phi_{\mu, \sigma^2}(-\lambda) - 1 \right) \\
    &+ \sigma^2 \left( \Phi'_{\mu, \sigma^2}(\lambda) - \Phi'_{\mu, \sigma^2}(-\lambda) \right) \\
    \mathsf{Var}(S_\lambda(Z)) =& 2(\mu^2 + \sigma^2 + \lambda^2) \\
    &- ((\mu + \lambda)^2 + \sigma^2) (1 - \Phi_{\mu, \sigma^2}(-\lambda)) \\
    &- ((\mu - \lambda)^2 + \sigma^2) \Phi_{\mu, \sigma^2}(\lambda) \\
    &- (\mu + \lambda) \sigma^2  \Phi'_{\mu, \sigma^2}(-\lambda) \\
    &+ (\mu - \lambda) \sigma^2  \Phi'_{\mu, \sigma^2}(\lambda) \\
    &- \E[S_\lambda(Z)]^2. 
\end{align}
\end{theorem}
\begin{proof}
    The cdf of $S_\lambda(Z)$ is given by
    \begin{align}
        f(x) = \begin{cases}
            \Phi_{\mu,\sigma^2}(x - \lambda), &x < 0, \\
            \Phi_{\mu,\sigma^2}(x + \lambda), &x > 0. \\
        \end{cases}
    \end{align}
    The mean of $S_\lambda(Z)$ is therefore given by
    \begin{align}
        \E[S_\lambda(Z)] &= \int_{-\infty}^{0} x \Phi'_{\mu,\sigma^2}(x - \lambda) dx + \int_0^{\infty} x \Phi'_{\mu,\sigma^2}(x + \lambda) dx.
    \end{align}
    The result follows by evaluating the above integral. To find the variance, we find $\E[S_\lambda(Z)^2]$, by which $\Var(S_\lambda(Z)) = \E[S_\lambda(Z)^2] - \E[S_\lambda(Z)]^2$. The expected value of $S_\lambda(Z)^2$ is given by the integral
    \begin{align}
        \E[S_\lambda(Z)^2] &= \int_{-\infty}^{0} x^2 \Phi'_{\mu,\sigma^2}(x - \lambda) dx + \int_0^{\infty} x^2 \Phi'_{\mu,\sigma^2}(x + \lambda) dx.
    \end{align}
    Evaluating the integral gives the result. 
\end{proof}
Note that as $\sigma^2 \to 0$, then
\begin{align}
    \E[S_\lambda(Z)] &\to \begin{cases}
        0, & \text{if } |\mu| \leq \lambda, \\
        \mu - \lambda, & \text{if } \mu > \lambda, \\
        \mu + \lambda, & \text{if } \mu < -\lambda,
    \end{cases} \\
    \Var[S_\lambda(Z)] &\to 0.
\end{align}
Furthermore, as $\sigma^2 \to 0$ and $\lambda \to 0$, $\E[S_\lambda(Z)] \to \mu$. 



\section{Simulation study 3: indirectly-standardized outcomes}
\label{sec:simulation-study-3}
For the third simulation study we used the data-generating process described as the second simulation study of \citep{susmann2024doublyrobustnonparametricefficient}, which we refer to for a detailed description. The number of providers was set to $m = 50$ and the number of covariates to $k = 5$ for all simulations. The data-generating process was sampled $250$ times for each sample size in $N \in \{3000, 5000, 10000 \}$. For each simulated dataset, the TMLE method described in \citep{susmann2024doublyrobustnonparametricefficient} was applied to estimate the indirectly standardized readmission ratio. Nuisance parameters were estimated using \texttt{lightgbm} with 200, 100, and 50 iterations, \texttt{glm}, and \texttt{gam} learners. These unpenalized estimates were then adjusted using our proposed $L2$ and $L1$ penalization approach with data-adaptive choice of tuning parameter. We also applied the Empirical Bayes adjustment described in Section~\ref{sec:ridge}. The results were compared by their mean squared error (MSE), mean error (ME), and empirical coverage of the 95\% confidence intervals.

Results are shown in Table~\ref{tab:simulation-indirectly-standardized-outcomes}. The Empirical Bayes adjustment achieved the lowest mean squared error, at the expense of having the highest bias. The $L_2$ penalized estimators had the second-lowest mean squared error while also having lower bias than the unpenalized estimates. The 95\% confidence intervals for all methods were anti-conservative, with $L_1$ penalized estimates in the smallest sample size exhibiting the worst empirical coverage.

\begin{table}[ht]
    \centering
    \small
    \begin{tabular}{|rrrrr|}
    \hline
    $N$ & $\psi$ & $L_1$ & $L_2$ & EB \\
    \hline
    \multicolumn{5}{|l|}{\textit{Mean Squared Error $\times 100$}} \\
    3000 & 18.4 & 17.5 & 10.1 & 4.9\\
    5000 & 8.7 & 8.4 & 5.5 & 3.2\\
    10000 & 3.0 & 3.0 & 2.6 & 1.8\\
    \multicolumn{5}{|l|}{\textit{Mean Error $\times 100$}} \\
    3000 & -5.5 & -4.4 & 1.9 & 7.8\\
    5000 & -2.6 & -2.5 & 1.8 & 6.5\\
    10000 & -1.4 & -1.4 & 0.9 & 3.8\\
    \multicolumn{5}{|l|}{\textit{95\% Empirical Coverage}} \\
    3000 & 92.0\% & 72.6\% & 84.5\% & 90.8\%\\
    5000 & 93.4\% & 89.8\% & 88.0\% & 91.7\%\\
    10000 & 93.7\% & 93.6\% & 90.9\% & 92.8\%\\
     \hline
    \end{tabular}
    \caption{Results from Simulation Study 3 comparing the unpenalized TMLE estimator, $L_1$-regularized estimator, $L_2$-regularized estimator, and Empirical Bayes (EB) shrinkage estimator.}
    \label{tab:simulation-indirectly-standardized-outcomes}
\end{table}

\section{Additional simulation results}
\label{appendix:additional-simulations}

\begin{table}[ht]
    \begin{tabular}{|lr|lllll|lllll|}
        \hline
        & & \multicolumn{5}{c}{MSE $\times 100$} & \multicolumn{5}{c|}{ME $\times 100$} \\
        $\sigma$ & $N$ & $\psi$ & $L_1$ & Lasso & $L_2$ & Ridge & $\psi$ & $L_1$ & Lasso & $L_2$ & Ridge\\
        \hline
        0.5 & 50 & 190.6 & 96.1 & 78.4 & 87.1 & 77.0 & -0.7 & -0.5 & -0.1 & -0.5 & -0.3\\
         & 100 & 51.2 & 39.2 & 16.6 & 34.8 & 66.3 & 0.2 & 0.0 & -0.4 & 0.1 & 0.2\\
         & 250 & 1.5 & 1.5 & 0.7 & 1.5 & 0.9 & 0.0 & 0.0 & 0.0 & 0.0 & 0.0\\
         & 500 & 0.4 & 0.4 & 0.3 & 0.4 & 0.4 & 0.0 & 0.0 & 0.0 & 0.0 & 0.0\\
         & 1000 & 0.1 & 0.1 & 0.1 & 0.1 & 0.2 & 0.0 & 0.0 & -0.1 & 0.0 & 0.0\\
        1 & 50 & 197.2 & 99.1 & 81.4 & 90.0 & 77.9 & -1.3 & -1.0 & -0.9 & -1.1 & -0.9\\
         & 100 & 61.2 & 45.3 & 25.3 & 40.1 & 66.8 & -0.8 & -0.5 & -0.2 & -0.6 & -0.3\\
         & 250 & 4.1 & 4.0 & 2.6 & 3.9 & 2.7 & 0.0 & 0.0 & 0.0 & 0.0 & 0.0\\
         & 500 & 1.3 & 1.3 & 1.0 & 1.3 & 1.1 & 0.0 & 0.0 & 0.0 & 0.0 & 0.0\\
         & 1000 & 0.5 & 0.5 & 0.5 & 0.5 & 0.5 & 0.0 & 0.0 & 0.0 & 0.0 & 0.0\\
        3 & 50 & 269.2 & 117.4 & 95.5 & 109.9 & 79.1 & -0.7 & 0.0 & 0.7 & 0.3 & 0.7\\
         & 100 & 122.4 & 72.2 & 65.4 & 64.4 & 67.4 & 0.3 & 0.1 & -0.3 & 0.0 & -0.4\\
         & 250 & 28.3 & 24.0 & 20.0 & 22.2 & 18.1 & -0.1 & -0.1 & -0.2 & -0.1 & -0.1\\
         & 500 & 10.1 & 9.5 & 8.4 & 9.2 & 8.1 & 0.0 & 0.0 & 0.1 & 0.1 & 0.1\\
         & 1000 & 4.2 & 4.1 & 3.9 & 4.1 & 3.8 & 0.1 & 0.1 & 0.0 & 0.1 & 0.1\\
        \hline
    \end{tabular}
    \newline
    \newline
    \newline
    \begin{tabular}{|lr|lllll|lllll|}
        \hline
        & & \multicolumn{5}{c}{Var $\times 100$} & \multicolumn{5}{c|}{95\% Empirical Coverage} \\
        $\sigma$ & $N$ & $\psi$ & $L_1$ & Lasso & $L_2$ & Ridge & $\psi$ & $L_1$ & Lasso & $L_2$ & Ridge\\
        \hline
        0.5 & 50   & 190.6 & 87.1  & 77.0 & 96.1  & 78.4 & 94.2\% & 98.5\% & & 90.9\% & \\
            & 100  & 51.2  & 34.8  & 66.3 & 39.2  & 16.6 & 94.5\% & 97.0\% & & 91.9\% & \\
            & 250  & 1.5   & 1.5   & 0.9  & 1.5   & 0.7  & 93.2\% & 93.2\% & & 93.1\% & \\
            & 500  & 0.4   & 0.4   & 0.4  & 0.4   & 0.3  & 93.5\% & 93.6\% & & 93.6\% & \\
            & 1000 & 0.1   & 0.1   & 0.2  & 0.1   & 0.1  & 94.0\% & 94.0\% & & 93.9\% & \\
        1   & 50   & 197.2 & 90.0  & 77.9 & 99.0  & 81.4 & 93.8\% & 98.6\% & & 90.5\% & \\
            & 100  & 61.2  & 40.1  & 66.8 & 45.3  & 25.3 & 94.3\% & 97.3\% & & 91.3\% & \\
            & 250  & 4.1   & 3.9   & 2.7  & 4.0   & 2.6  & 94.4\% & 94.5\% & & 93.8\% & \\
            & 500  & 1.3   & 1.3   & 1.1  & 1.3   & 1.0  & 94.2\% & 94.2\% & & 94.1\% & \\
            & 1000 & 0.5   & 0.5   & 0.5  & 0.5   & 0.5  & 94.8\% & 94.7\% & & 94.7\% & \\
        3   & 50   & 269.2 & 109.9 & 79.1 & 117.4 & 95.5 & 94.2\% & 98.9\% & & 91.4\% & \\
            & 100  & 122.4 & 64.4  & 67.4 & 72.2  & 65.4 & 94.5\% & 98.3\% & & 91.4\% & \\
            & 250  & 28.3  & 22.2  & 18.1 & 24.0  & 20.0 & 94.6\% & 96.3\% & & 92.6\% & \\
            & 500  & 10.1  & 9.2   & 8.1  & 9.5   & 8.4  & 94.8\% & 95.5\% & & 93.8\% & \\
            & 1000 & 4.2   & 4.1   & 3.8  & 4.1   & 3.9  & 95.2\% & 95.4\% & & 94.7\% & \\
        \hline
    \end{tabular}
    \caption{Results from Simulation Study 1 for non-parametric linear association parameters comparing mean squared error (MSE), mean error (ME), variance (Var), and 95\% empirical coverage. The estimators considered are the unpenalized estimates, $L_1$-penalized estimates, and $L_2$-penalized estimates. As a benchmark, results for penalized linear regression with $L_1$ (Lasso) and $L_2$ (Ridge) penalties are shown. The simulations have varying outcome noise standard deviations $\sigma$ and overall sample sizes $N$.}
    \label{tab:simulation-covariance}
\end{table}

\begin{table}[ht]
    \centering
    \small
    \begin{tabular}{|lrrrrrrrrrrrrr|}
    \hline
    \multicolumn{2}{|l}{} & \multicolumn{4}{c}{MSE $\times 100$} & \multicolumn{4}{c}{ME $\times 100$} & \multicolumn{4}{c|}{95\% Empirical Coverage} \\
    $\sigma$ & $N$ & $\psi_n$ & $L_1$ & $L_2$ & EB & $\psi_n$ & $L_1$ & $L_2$ & EB & $\psi_n$ & $L_1$ & $L_2$ & EB \\
    \hline
    \multicolumn{14}{|l|}{$\theta = 0\%$} \\
    0.5 & 4000 & 0.8 & 0.2 & 0.2 & 0.2 & -0.1 & 0.0 & 0.0 & 0.0 & 94.3\% & 99.2\% & 94.3\% & 94.3\%\\
       & 6000 & 0.5 & 0.1 & 0.1 & 0.1 & 0.0 & 0.0 & 0.0 & 0.0 & 94.4\% & 99.4\% & 94.4\% & 94.4\%\\
       & 8000 & 0.4 & 0.1 & 0.1 & 0.1 & -0.2 & -0.1 & -0.1 & -0.1 & 94.5\% & 99.4\% & 94.5\% & 94.5\%\\
       & 10000 & 0.3 & 0.1 & 0.1 & 0.1 & -0.1 & -0.1 & -0.1 & 0.0 & 94.5\% & 99.3\% & 94.5\% & 94.5\%\\
      1 & 4000 & 3.1 & 0.8 & 0.9 & 0.9 & -0.2 & -0.1 & -0.1 & -0.1 & 93.6\% & 98.9\% & 93.6\% & 93.6\%\\
       & 6000 & 1.9 & 0.4 & 0.5 & 0.5 & -0.1 & 0.0 & 0.0 & 0.0 & 95.0\% & 99.4\% & 95.0\% & 95.0\%\\
       & 8000 & 1.5 & 0.3 & 0.4 & 0.4 & -0.2 & -0.1 & -0.1 & -0.1 & 94.5\% & 99.1\% & 94.5\% & 94.5\%\\
       & 10000 & 1.2 & 0.3 & 0.3 & 0.3 & -0.2 & -0.1 & -0.1 & -0.1 & 94.9\% & 99.3\% & 94.9\% & 94.9\%\\
      2 & 4000 & 12.1 & 2.9 & 3.3 & 3.5 & -1.5 & -0.6 & -0.8 & -0.8 & 94.1\% & 99.1\% & 94.1\% & 94.1\%\\
       & 6000 & 8.0 & 1.9 & 2.2 & 2.3 & -0.9 & -0.3 & -0.5 & -0.5 & 94.7\% & 99.2\% & 94.7\% & 94.7\%\\
       & 8000 & 5.8 & 1.3 & 1.5 & 1.6 & -0.1 & 0.0 & 0.0 & 0.0 & 94.8\% & 99.3\% & 94.8\% & 94.8\%\\
       & 10000 & 4.6 & 1.1 & 1.2 & 1.3 & -0.4 & -0.1 & -0.2 & -0.2 & 94.7\% & 99.3\% & 94.7\% & 94.7\%\\
      4 & 4000 & 48.8 & 12.4 & 13.8 & 14.3 & -2.7 & -1.1 & -1.5 & -1.6 & 93.8\% & 98.9\% & 93.8\% & 93.8\%\\
       & 6000 & 30.4 & 6.7 & 8.0 & 8.3 & -1.5 & -0.7 & -0.8 & -0.8 & 94.9\% & 99.4\% & 94.9\% & 94.9\%\\
       & 8000 & 23.3 & 5.4 & 6.3 & 6.5 & -0.5 & -0.1 & -0.3 & -0.3 & 94.6\% & 99.3\% & 94.6\% & 94.6\%\\
       & 10000 & 19.2 & 4.7 & 5.3 & 5.6 & -0.6 & -0.1 & -0.2 & -0.2 & 94.3\% & 99.2\% & 94.3\% & 94.3\%\\
    \multicolumn{14}{|l|}{$\theta = 100\%$} \\
    0.5 & 4000 & 0.8 & 0.7 & 0.7 & 0.7 & -0.1 & -0.1 & -0.1 & 0.0 & 94.3\% & 94.2\% & 93.9\% & 93.9\%\\
       & 6000 & 0.5 & 0.5 & 0.5 & 0.5 & -0.2 & -0.2 & -0.2 & -0.1 & 94.9\% & 94.9\% & 94.6\% & 94.6\%\\
       & 8000 & 0.4 & 0.4 & 0.4 & 0.4 & -0.1 & -0.1 & -0.1 & 0.0 & 94.6\% & 94.7\% & 94.4\% & 94.4\%\\
       & 10000 & 0.3 & 0.3 & 0.3 & 0.3 & 0.0 & 0.0 & 0.0 & 0.1 & 94.8\% & 94.8\% & 94.8\% & 94.8\%\\
      1 & 4000 & 3.0 & 2.9 & 2.7 & 2.7 & -0.6 & -0.6 & -0.5 & -0.2 & 94.4\% & 94.8\% & 93.6\% & 93.8\%\\
       & 6000 & 2.0 & 1.9 & 1.9 & 1.9 & -0.3 & -0.3 & -0.3 & -0.1 & 94.2\% & 94.6\% & 93.6\% & 93.8\%\\
       & 8000 & 1.5 & 1.4 & 1.4 & 1.4 & 0.1 & 0.1 & 0.1 & 0.3 & 94.5\% & 94.6\% & 93.9\% & 93.8\%\\
       & 10000 & 1.1 & 1.1 & 1.1 & 1.1 & -0.3 & -0.3 & -0.3 & -0.2 & 94.9\% & 94.8\% & 94.2\% & 94.3\%\\
      2 & 4000 & 11.9 & 10.4 & 9.0 & 9.0 & -0.8 & -0.8 & -0.8 & -0.3 & 94.5\% & 96.0\% & 91.7\% & 92.1\%\\
       & 6000 & 8.1 & 7.5 & 6.7 & 6.7 & -0.5 & -0.4 & -0.3 & 0.2 & 94.1\% & 95.2\% & 91.8\% & 91.9\%\\
       & 8000 & 5.8 & 5.5 & 5.0 & 5.0 & -1.1 & -1.1 & -1.0 & -0.6 & 94.7\% & 95.4\% & 92.5\% & 92.7\%\\
       & 10000 & 4.6 & 4.3 & 4.0 & 4.1 & -0.4 & -0.4 & -0.3 & 0.1 & 94.6\% & 95.5\% & 93.2\% & 93.5\%\\
      4 & 4000 & 49.0 & 29.9 & 25.1 & 25.5 & -2.2 & -1.4 & -1.0 & -0.3 & 94.0\% & 98.5\% & 90.5\% & 90.9\%\\
       & 6000 & 31.0 & 21.7 & 18.2 & 18.5 & -1.8 & -2.0 & -1.7 & -0.9 & 94.6\% & 98.5\% & 91.0\% & 91.4\%\\
       & 8000 & 23.2 & 17.8 & 14.8 & 14.9 & -0.3 & -0.2 & -0.1 & 0.6 & 94.3\% & 97.9\% & 91.1\% & 91.2\%\\
       & 10000 & 18.3 & 15.0 & 12.6 & 12.8 & -1.6 & -1.4 & -1.3 & -0.6 & 95.0\% & 97.2\% & 91.3\% & 91.4\%\\
    \hline
    \end{tabular}
    \caption{Additional results from Simulation Study 2 for group-specific ATEs showing mean squared error (MSE), mean error (ME), and empirical 95\% coverage for the unpenalized TMLE estimator ($\psi_n$), $L_1$ penalized parameter, and $L_2$ penalized parameter for varying probabilities of positive group-specific treatment effect $\theta$, outcome noise standard deviations $\sigma$, and overall sample sizes $N$.}
    \label{tab:simulation-group-effects-supplemental}
\end{table}
\end{document}

%% file: _preamble.tex
\usepackage{graphicx} 
\usepackage{amsmath}
\usepackage{amsfonts}
\usepackage{amssymb}
\usepackage{mathtools}
\usepackage[margin=1in]{geometry}
\usepackage{natbib}
\usepackage{amsthm}
\usepackage{authblk}
\usepackage{hyperref}
 \usepackage[T1]{fontenc}
\usepackage[english]{babel}
\usepackage{accents}
\usepackage{pdflscape}
\usepackage{chngcntr}
\hypersetup{
    colorlinks=true,
    linkcolor=red,
    filecolor=magenta,
    citecolor=blue,
    urlcolor=blue,
}
\usepackage{appendix}
\usepackage[most]{tcolorbox}

\providecommand{\keywords}[1]{\textbf{\textit{Keywords---}} #1}

\mathtoolsset{showonlyrefs}

\newtheorem{assumption}{Assumption}
\DeclareMathOperator*{\argmin}{\arg\!\min}

\newcounter{example}
\def\exampletext{Example \thetcbcounter} 

\NewDocumentEnvironment{example}{ O{} }
{
\colorlet{colexam}{red!55!black} 
\newtcolorbox[use counter=example]{examplebox}{%
    empty,
    title={\exampletext: #1},
    attach boxed title to top left,
       minipage boxed title,
    boxed title style={empty,size=minimal,toprule=0pt,top=4pt,left=3mm,overlay={}},
    coltitle=colexam,fonttitle=\bfseries,
    before=\par\medskip\noindent,parbox=false,boxsep=0pt,left=3mm,right=0mm,top=2pt,breakable,pad at break=0mm,
       before upper=\csname @totalleftmargin\endcsname0pt, 
    overlay unbroken={\draw[colexam,line width=.5pt] ([xshift=-0pt]title.north west) -- ([xshift=-0pt]frame.south west); },
    overlay first={\draw[colexam,line width=.5pt] ([xshift=-0pt]title.north west) -- ([xshift=-0pt]frame.south west); },
    overlay middle={\draw[colexam,line width=.5pt] ([xshift=-0pt]frame.north west) -- ([xshift=-0pt]frame.south west); },
    overlay last={\draw[colexam,line width=.5pt] ([xshift=-0pt]frame.north west) -- ([xshift=-0pt]frame.south west); },%
    }
\begin{examplebox}}
{\end{examplebox}\endlist}

\newcommand{\E}{\mathsf{E}}
\newcommand{\Var}{\mathsf{Var}}
\newcommand{\Cov}{\mathsf{Cov}}
\newcommand{\MSE}{\mathsf{MSE}}
\newcommand{\Bias}{\mathsf{Bias}}
\newcommand{\Variance}{\mathsf{Variance}}
\newcommand{\I}{\mathbb{I}}

\newcommand{\propscore}{\pi}
\newcommand{\outcomemodel}{\mu}

\newcommand{\loss}{L}
\newcommand{\penaltyf}{V}
\newcommand{\objective}{U}

\newcommand{\param}{\psi}
\newcommand{\penalparam}{\tilde{\psi}}
\newcommand{\model}{\mathcal{M}}
\newcommand{\paramdim}{|\mathcal{D}|}

\newcommand{\eif}{D^*}
\newcommand{\influencef}{D}

\newcommand{\convd}{\overset{d}{\rightarrow}}
\newcommand{\convp}{\overset{p}{\rightarrow}}

\newcommand{\os}{\mathsf{os}}

\newcommand{\criterion}{\mathsf{Crit}}
